\documentclass[10pt,twocolumn,letterpaper]{article}

\usepackage{cvpr}
\usepackage{times}
\usepackage{epsfig}
\usepackage{graphicx}
\usepackage{amsmath}
\usepackage{amssymb}
\usepackage{amsthm}
\usepackage{color}
\usepackage{algorithm}
\usepackage{algpseudocode}

\newtheorem{thm}{Theorem}
\newtheorem{prop}{Proposition}
\newtheorem{exmp}{Example}
\newtheorem{rmks}{Remark}

\newcommand{\R}{\mathbb{R}}
\newcommand{\Z}{\mathbb{Z}}
\newcommand{\N}{\mathbb{N}}

\newcommand{\rmdim}{\mathrm{dim}}
\newcommand{\ect}{\mathrm{ECT}}
\newcommand{\wect}{\mathrm{WECT}}

\usepackage{hyperref}

\cvprfinalcopy % *** Uncomment this line for the final submission

\def\cvprPaperID{****} % *** Enter the CVPR Paper ID here

\ifcvprfinal\fi
\begin{document}

%%%%%%%%% TITLE
\title{The Weighted Euler Curve Transform for Shape and Image Analysis}

\author{
Qitong Jiang \\
The Ohio State University \\
Department of Mathematics\\
{\tt\small jiang.927@osu.edu}
\and
Sebastian Kurtek\\
The Ohio State University\\
Department of Statistics \\
{\tt\small kurtek.1@stat.osu.edu}
\and
Tom Needham \\
Florida State University \\
Department of Mathematics\\
{\tt\small tneedham@fsu.edu}
}

\maketitle
%\thispagestyle{empty}

%%%%%%%%% ABSTRACT
\begin{abstract}
    The Euler Curve Transform (ECT) of Turner et al.\ is a complete invariant of an embedded simplicial complex, which is amenable to statistical analysis. We generalize the ECT to provide a similarly convenient representation for weighted simplicial complexes, objects which arise naturally, for example, in certain medical imaging applications. We leverage work of Ghrist et al.\ on Euler integral calculus to prove that this invariant---dubbed the Weighted Euler Curve Transform (WECT)---is also complete. We explain how to transform a segmented region of interest in a grayscale image into a weighted simplicial complex and then into a WECT representation. This WECT  representation is applied to study Glioblastoma Multiforme brain tumor shape and texture data. We show that the WECT representation is effective at clustering tumors based on qualitative shape and texture features and that this clustering correlates with patient survival time.
\end{abstract}

%%%%%%%%% BODY TEXT
\section{Introduction}

Tools from algebraic topology have become increasingly popular in shape analysis applications over the past several years. At an intuitive level, the topological perspective is appealing because algebraic topology is, at its core, designed to extract tractable algebraic invariants from complex shape data. The dominant technique in topological shape analysis is \emph{persistent homology}, which summarizes multiscale topological features of a shape, where scale is measured relative to some \emph{filtration function}. Roughly, for a continuous function $f:X \rightarrow \R$ on a topological space $X$ (satisfying certain tameness conditions), one computes the degree-$k$ homology of the sublevel sets $f^{-1}((-\infty,r])$ and tracks ``births" and ``deaths" of homological features as the filtration value $r$ is increased. This produces a summary statistic for the pair $(X,f)$ called a \emph{persistence diagram} (see standard references \cite{edelsbrunner2010computational,carlsson2014topological}), which can be used as a proxy for $X$ in shape analysis applications. This approach has been taken in several shape analysis tasks, with shape data coming from cortical surfaces \cite{chung2009persistence}, brain artery systems \cite{bendich2016persistent}, proteins \cite{kovacev2016using} and leaf contours \cite{patrangenaru2018challenges}. 
While the persistence diagram of a pair $(X,f)$ provides a computationally tractable shape summary, the complex structure of the invariant means that it is difficult to incorporate into statistical models. A simpler invariant is the \emph{Euler curve} of $(X,f)$; this is an integer-valued function on $\R$ whose value at $r$ is the Euler characteristic (i.e., the alternating sum of ranks of the homology groups) of the sublevel set $f^{-1}((-\infty,r])$. 

Given shape data, one must answer the question of which filtration function to apply in order to apply these topological methods. For a shape represented as a simplicial complex $K$ embedded in a Euclidean space $\R^d$, recent work has advocated for using an ensemble of filtration functions given by the height function along directions sampled from the unit sphere $S^{d-1}$ \cite{turner2014persistent,ghrist2018persistent,fasy2018challenges,curry2018many,betthauser2018topological,crawford2019predicting,fasy2019persistence}. The collection of all persistence diagrams for these height filtrations is referred to as the persistent homology transform of $K$. Likewise, the collection of Euler curves for all filtration directions is called the Euler curve transform (ECT) for $K$. The ECT provides a particularly attractive shape representation, as its simplistic structure allows it to be easily incorporated into statistical models. This was the approach taken in \cite{crawford2019predicting}, where the ECTs for Glioblastoma Multiforme (GBM) brain tumor shapes were used as covariates in a model for survival prediction. 

In this paper, we consider a variant of the ECT, which we dub the \emph{weighted Euler Characteristic Transform} (WECT). This object is defined for shape data consisting of an embedded simplicial complex $K$ endowed with an extra weighting function $g$. The pair $(K,g)$ is referred to as a \emph{weighted simplicial complex}. The WECT invariant incorporates both the shape of $K$ and the weighting function $g$ into a topological summary. Our motivation for defining this summary also comes from analysis of brain tumor data, which is naturally given as a segmented grayscale image. The segmented shape is used to construct a simplicial complex $K$ embedded in $\R^2$, and the grayscale pixel values inside the shape define the weight function $g$. While the WECT is a simple generalization of the ECT, it is able to efficiently incorporate vital information that is ignored by the ECT.

\subsection{Contributions and Organization of Paper}

The proposed mathematical framework is laid out in detail in Section \ref{sec:mathematical_framework}. There, we give a precise definition of the WECT as a generalization of ideas appearing in \cite{turner2014persistent,betthauser2018topological}. We show that recent work of Ghrist, Levanger and Mai implies that the WECT is a complete descriptor of weighted simplicial complexes, i.e., two weighted simplicial complexes have the same WECT if and only if they are equal. In this section, we also provide comparisons between the WECT and other techniques appearing in the topological shape analysis literature. In Section \ref{sec:applications}, we demonstrate some applications of the WECT framework. We begin with a toy example exploring the utility of the WECT in classifying and registering MNIST digit images. Next, we explore a real application wherein we study the shape and appearance of Glioblastoma Multiforme tumors using WECT representations. Using a simple distance-based clustering scheme, we are able to distinguish clusters of tumors with low survival times, purely from imaging data. Open source code for producing and analyzing WECTs has been made publicly available \cite{code}.

\section{Mathematical Framework}\label{sec:mathematical_framework}

In this section, we lay out the mathematical framework for the WECT. We begin by reviewing some basic definitions in order to set notation.

\subsection{Simplicial Complexes and the Euler Characteristic}

Let $K$ be a \emph{simplicial complex} embedded in some Euclidean space $\R^d$. That is, $K$ is a set of embedded \emph{simplices} $\sigma$. Each $\sigma$ is the convex hull of a set of $k+1$ points in general position in $\R^d$, where $k \leq d$ is the \emph{dimension} of the simplex; we write $k = \rmdim(\sigma)$. For example, a $0$-dimensional simplex is a point, a $1$-dimensional simplex is a closed line segment and a $2$-dimensional simplex is a triangle. The $k$ points defining $\sigma$ are called its \emph{vertices}. The convex hull of $\ell < k$ of these vertices is also a simplex of $K$ and is called an \emph{$\ell$-dimensional face} of $\sigma$. If $\tau$ is a face of $\sigma$, we write $\tau < \sigma$. If $\sigma$ and $\tau$ are simplices of $K$, we require that $\sigma \cap \tau$ is also a simplex of $K$. The maximum dimension of a simplex in $K$ is called the \emph{dimension} of $K$, denoted $\rmdim(K)$. A collection of simplices of $K$ which itself forms a simplicial complex is called a \emph{subcomplex} of $K$. The union of all simplices of $K$ of dimension less than or equal to $\ell$ is a subcomplex called the \emph{$\ell$-skeleton} of $K$, denoted $K^{ \leq \ell}$. The set of simplices of $K$ of dimension exactly $\ell$ is denoted $K^{\ell}$; note that $K^{\ell}$ is not a simplicial complex in general.

Abusing notation, we will alternate between treating each embedded simplicial complex as a combinatorial object (a set of simplices) and as a geometric object (a set of points in $\R^d$). We hope that the interpretation should always be clear from context.

\begin{figure}
    \centering
    \includegraphics[width = 0.22\textwidth]{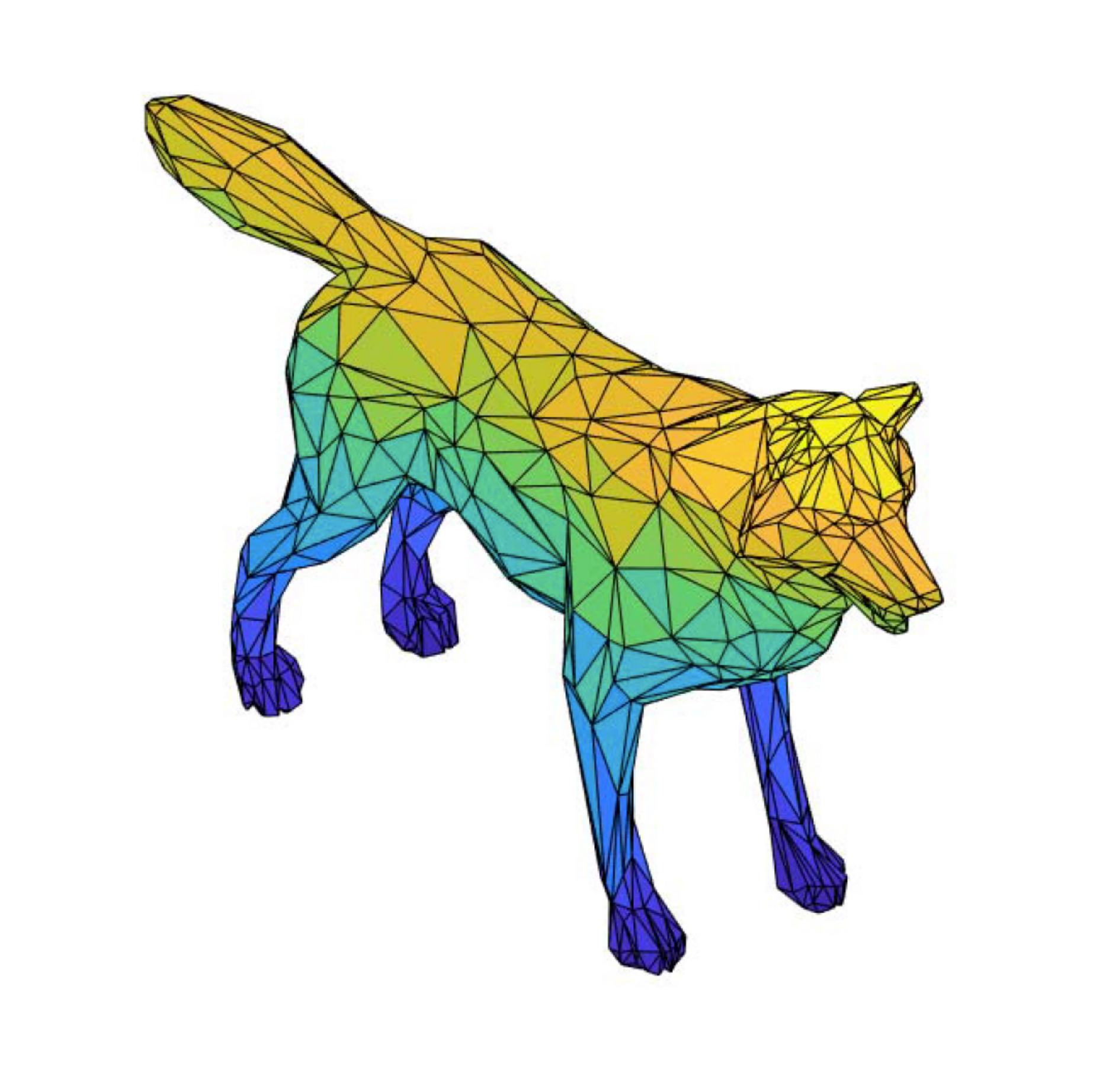}
    \includegraphics[width = 0.18\textwidth]{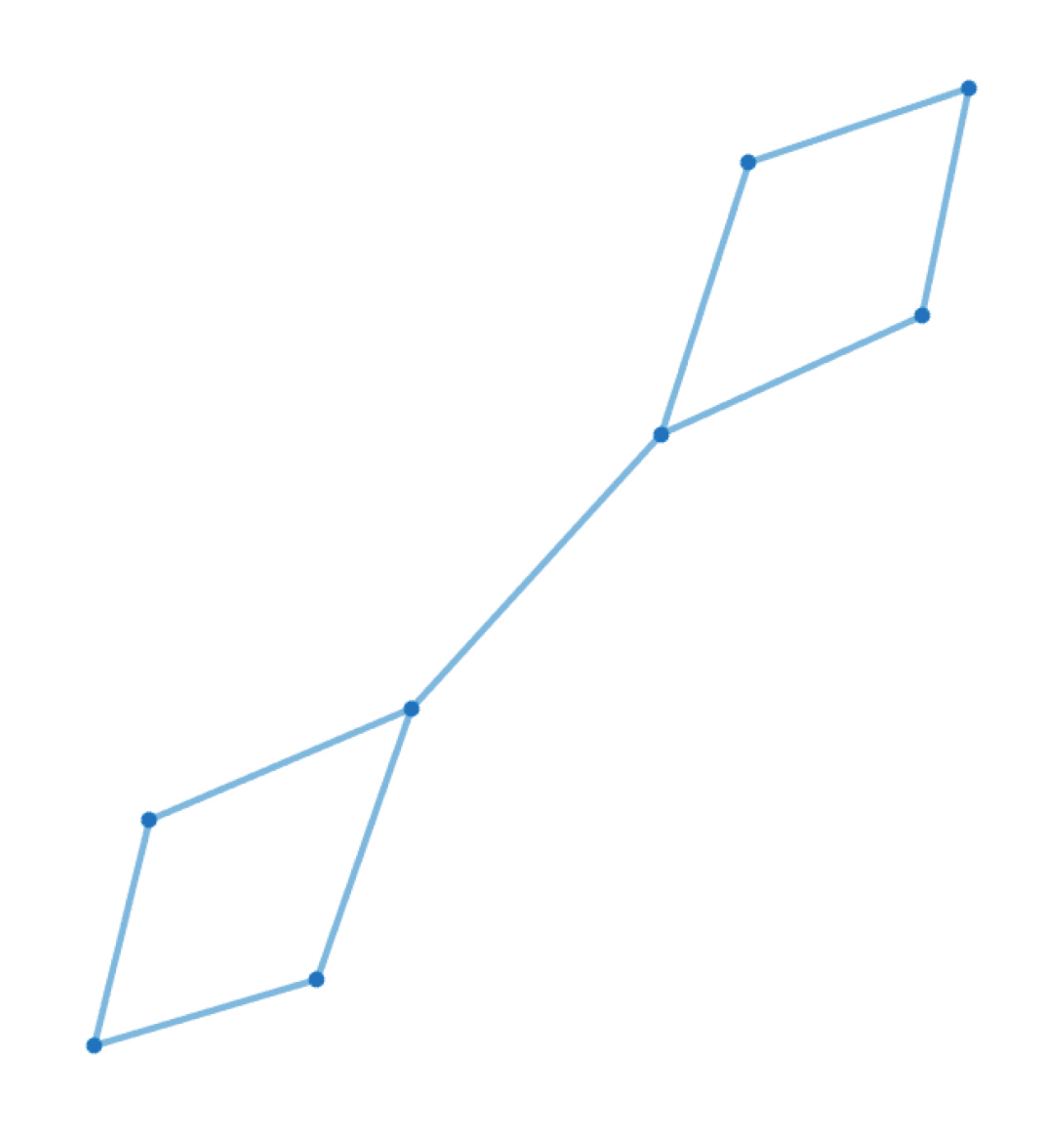}
    \caption{Examples of embedded simplicial complexes commonly arising in computer vision. A triangulated surface is a two-dimensional simplicial complex embedded in $\R^3$. An embedded planar graph is a 1-dimensional simplicial complex in $\R^2$.}
    \label{fig:simplicial_complexes}
\end{figure}
A simple combinatorial invariant of a simplicial complex is its \emph{Euler characteristic}, denoted $\chi(K)$. The Euler characteristic is defined as
\begin{equation*}
\chi(K) = \sum_{d = 0}^{\rmdim(K)} (-1)^{d} \cdot  \# K^d,
\end{equation*}
where $\# A$ will generally be used to denote the cardinality of a set $A$. The concept of the Euler characteristic generalizes to more flexible classes of spaces, and it is a basic fact of algebraic topology that $\chi$ is a homotopy equivalence invariant. Simplicial complexes form a convenient category for computation, since they can be represented abstractly in a purely combinatorial way by keeping track of all simplices and their inclusions. In this paper, we are focused on the geometrically motivated case where are simplicial complexes are specified by an embedding into a Euclidean space. While not strictly necessary, the invariants we describe are most interesting when $K \subset \R^d$ is a $d$-dimensional simplicial complex. Moreover, we restrict our attention to the finite setting, i.e., $\# K^\ell$ is finite for all $\ell$.

\subsection{Euler Curve Transform}

Consider a function $f:K \rightarrow \R$ as an assignment of a real number to each simplex of $K$, i.e., the function is constant along faces. The function is a \emph{filtration function} if each sublevel set $f^{-1}((-\infty, r])$ is a subcomplex of $K$. A filtration function induces a chain of inclusions of simplicial complexes $f^{-1}((-\infty,r_1]) \subset f^{-1}((-\infty,r_2]) \subset \cdots \subset f^{-1}((-\infty,r_n])$, where $r_1 < r_2 < \cdots < r_n$ are the finitely many (using the assumption that $K$ is finite) values in the range of $f$. From this data, one obtains the \emph{Euler curve} $\chi_f:\R \rightarrow \Z$ defined as $\chi_f(r) = \chi \left(f^{-1}((-\infty,r])\right)$.

\begin{figure*}[!t]
    \centering
    \includegraphics[width = 0.75\textwidth]{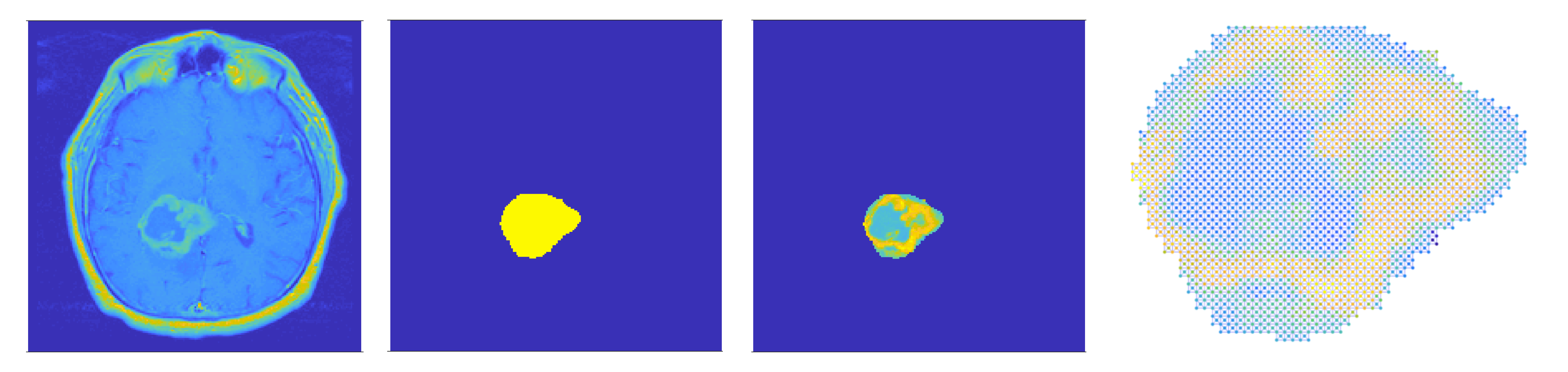}
    \caption{Glioblastoma multiforme tumor image data. From left to right: axial slice with largest tumor area selected from a 3D MRI image; binary tumor segmentation mask; segmented tumor image; weighted simplicial complex created from the segmented tumor image. Observe that the tumor shape data from the segmentation mask is enriched by the overlaid pixel value function extracted from the original image: the level sets of the pixel value function have interesting shape and topological features.}
    \label{fig:tumor_images}
\end{figure*}

Given data consisting of an embedded simplicial complex and a relevant function (or more general space and function where similar concepts can be defined), the Euler curve produces a multiscale topological summary which is amenable to classical analysis, and can be viewed as a simplification of the richer but more computationally taxing persistence diagram \cite{edelsbrunner2010computational,carlsson2014topological}. On the other hand, if a relevant function is not provided, one is left with the question of how to filter the simplicial complex.

It was observed in \cite{turner2014persistent} that for an embedded complex $K \subset \R^d$, there is a family of natural filtration functions: orthogonal projections onto the oriented one-dimensional subspaces of $\R^d$, which can be parameterized by the unit sphere $S^{d-1} \subset \R^d$. The \emph{Euler Curve Transform (ECT)} of an embedded simplicial complex $K \subset \R^d$ is the function $\ect_K: S^{d-1} \times \R \rightarrow \Z$ defined as
\begin{equation*}
\ect_K(v,r) = \chi_{p_v}(r),
\end{equation*}
with $p_v: K \rightarrow \R$ defined on the vertex set $K^0$ by the dot product
\begin{equation}\label{eqn:projection_1}
p_v(\sigma) = v \cdot \sigma.
\end{equation}
The function is extended inductively to higher-dimensional simplices as
\begin{equation}\label{eqn:projection_2}
p_v(\sigma) = \max \{p_v(\tau) \mid \tau < \sigma\}.
\end{equation}
In practical computations, one uses an approximation of the ECT given by sampling finitely many projection directions from $S^{d-1}$ and finitely many filtration values from $\R$.

One can also apply a smoothing operator to each single variable function $\ect_K(v,\cdot)$ to obtain the \emph{Smooth Euler Curve Transform (SECT)}. The SECT was applied in \cite{crawford2019predicting} to study Glioblastoma Multiforme tumor imaging data. In particular, the SECT served as a shape covariate in a Gaussian process regression model for survival prediction. Another variant of the ECT---very closely related to the one that we consider in subsequent sections---was applied in \cite{betthauser2018topological} to provide a topological signature for grayscale image data.

\subsection{Weighted Euler Characteristic}

Next, suppose that our data consists of an embedded simplicial complex $K \subset \R^d$ together with a function $g:K \rightarrow \N$, where $\N = \{1,2,\ldots\}$. We refer to the pair $(K,g)$ as a \emph{weighted simplicial complex}. The goal is to define a variant of $\ect_K$, which also incorporates data from $g$. We note that weighted simplicial complexes have already appeared in the literature in various contexts. To the best of our knowledge, they were first studied in \cite{dawson1990homology}, where a homology theory was developed. Abstract weighted simplicial complexes, i.e., those which do not come with a preferred embedding into a Euclidean space, serve as models for collaboration networks \cite{carstens2013persistent} and Vietoris-Rips complexes for weighted point clouds \cite{ren2018weighted}. We provide some examples of \emph{embedded} weighted simplicial complexes next.
\begin{exmp}\label{exmp:greyscale_images}
Our main motivating example comes from grayscale images containing a region of interest, e.g., a tumor image with a segmentation mask, which can be converted into weighted simplicial complexes using Algorithm \ref{alg:greyscale}. An example of this process is described in Figure \ref{fig:tumor_images}.
\end{exmp}

\begin{exmp}
Although the main examples considered in this paper will be of the form described in Example \ref{exmp:greyscale_images}, we note that there are many other situations where one might wish to consider weighted simplicial complexes. Given shape data as a simplicial complex $K$, one could consider the weight function $g$ as an annotation or measure of importance. For example, if $K$ is a complex representing a molecule shape, the weight function could be used to annotate different atom types. If $K$ is an anatomical surface, $g$ can be used to indicate regions of importance landmarked by a radiologist.
\end{exmp}

\begin{algorithm}[!t]
\caption{Grayscale Image to Weighted Complex}\label{alg:logmap}
\label{alg:greyscale}
\begin{algorithmic}[1]
\Function{ImageToWeightedComplex}{$A$}
\State \Comment{$A \in \N^{m \times n}$ greyscale image matrix} 
\State $V_{center} = \Call{find}{A \neq 0}$ \\
 \Comment{treat nonzero pixels as coords for vertices}
\State{$V = V_{center}$} \Comment{initialize vertex list}
\For{$v \in V_{center}$}\Comment{add corner vertices}
    \State append $v + [\pm 1/2, \pm 1/2]$ to $V$
\EndFor
\State{$V = \Call{unique}{V}$} \Comment{remove duplicates}
\State{$F = []$} \Comment{initialize face list}
\For{$v \in V_{center}$}
\State{append triangles containing $v$ to $F$}
\EndFor
\State{$E=$ all resulting edges}
\For{$f \in F$ containing $v \in V_{center}$}
\State{$Fw(f)=$ weight of corresponding pixel value} %\State{This defines list of face weights $Fw$.}
\EndFor
\For{$v \in V$}
\State{$Vw(v)=$ largest weight of face containing $v$} %\State{This defines list of vertex weights $Vw$.}
\EndFor
\For{$e \in E$}
\State{$Ew(e)=$ largest weight of face containing $e$} %\State{This defines list of edge weights $Ew$.}
\EndFor
\State{\Return $V, E, F, Vw, Ew, Fw$}
\EndFunction
\end{algorithmic}
\end{algorithm}

For a simplicial complex $K$ and a function $g:K \rightarrow \N$, we define the \emph{weighted Euler characteristic}
$$
\chi^w(K,g) = \sum_{d = 0}^{\rmdim(K)} (-1)^d \sum_{\sigma \in K^d} g(\sigma).
$$
\begin{rmks}
If $g(\sigma) = 1$ for all $\sigma \in K$, then $\chi^w(K,g) = \chi(K)$. The weighted Euler characteristic is therefore a direct generalization of the classical version.
\end{rmks}
\begin{rmks}
The same definition essentially appears in \cite{betthauser2018topological}; the only difference is that only simplicial complexes which are finite axis-aligned lattices were considered there.
\end{rmks}
\begin{rmks}
A generalization of the weighted Euler characteristic is a classical object of study in algebraic geometry; see, e.g.,  \cite{kashiwara1985index}.
\end{rmks}

We are particularly interested in functions $g:K \rightarrow \N$ which satisfy the consistency condition $g(\tau) = \max \{g(\sigma) \mid \tau < \sigma\}.$
Note that this condition is satisfied by the construction given in Algorithm \ref{alg:greyscale}. If a function satisfies this condition, we say that it is \emph{admissible}. For functions of this type, the weighted Euler characteristic has a natural interpretation. 

\begin{prop}
Suppose that $g:K \rightarrow \N$ is an admissible function. Then, each superlevel set $g^{-1}([z,\infty))$ is a subcomplex of $K$. The weighted Euler characteristic $\chi(K,g)$ is the sum of Euler characteristics of all superlevel complexes of $g$; that is,
\begin{equation}\label{eqn:weighted_euler_characteristic}
\chi^w(K,g) = \sum_{z \in \N} \chi(g^{-1}([z,\infty))).
\end{equation}
\end{prop}

\begin{proof}
We first show that the superlevel sets are subcomplexes of $K$. It suffices to show that for any $\sigma \in g^{-1}([z,\infty))$ and $\tau < \sigma$, we have $\tau \in g^{-1}([z,\infty))$. This is easy to see from the definition of an admissible function, since $\tau < \sigma$ implies $g(\tau) \geq g(\sigma) \geq z$, which implies $\tau \in g^{-1}([z,\infty))$. It remains to show that Equation \eqref{eqn:weighted_euler_characteristic} is true. In what follows, for a logical statement $S$, let $\textbf{1}_S$ denote the indicator function taking the value $1$ if $S$ is true, and $0$ if $S$ is false. Then,
\begin{align*}
   &\sum_{z \in \N} \chi(g^{-1}([z,\infty)))
   \\
   & \qquad = \sum_{z \in \N} \sum_{d=0}^{\rmdim(K)} (-1)^d \# \{ \sigma \in K \mid g(\sigma) \geq z \}^d \\
   & \qquad = \sum_{d = 0}^{\rmdim(K)} (-1)^d \sum_{z \in \N} \# \{ \sigma \in K^d \mid g(\sigma) \geq z\} \\
   & \qquad = \sum_{d = 0}^{\rmdim(K)} (-1)^d \sum_{z \in \N} \sum_{\sigma \in K^d} \textbf{1}_{g(\sigma) \geq z} \\
   & \qquad = \sum_{d = 0}^{\rmdim(K)} (-1)^d \sum_{\sigma \in K^d} g(\sigma) = \chi^w(K,g).
\end{align*}
\end{proof}

\subsection{Weighted Euler Curve Transform}

We now define the \emph{Weighted Euler Curve Transform (WECT)} as a straightforward generalization of the ECT; the WECT is specifically designed to treat weighted simplicial complexes. Let $(K,g)$ be a weighted simplicial complex, and let $f:K \rightarrow \R$ be a filtration function. The \emph{weighted Euler curve} associated to $f$ is the function $\chi^w_f:\R \rightarrow \Z$ defined as
$$
\chi^w_f(r) = \chi^w(f^{-1}((-\infty,r]),g),
$$
where $g$ is understood by context to be the restriction of $g$ to the subcomplex $f^{-1}((-\infty,r])$. We then define the WECT of a weighted simplicial complex $(K,g)$ with $K \subset \R^d$ as the function $\wect_{K,g}:S^{d-1} \times \R \rightarrow \Z$ defined as
$$
\wect_{K,g} (v,r) = \chi^w_{p_v}(r),
$$
with $p_v$ the projection function as defined in Equations \eqref{eqn:projection_1} and \eqref{eqn:projection_2}. Clearly, if the weight function $g$ is constant and equal to one, then $\wect_{K,g} = \ect_K$. 

As in the case of the ECT, a WECT is represented in practice by sampling a finite number of directions on the sphere $S^{d-1}$. An example of a WECT is shown in Figure \ref{fig:MNIST_WECT}. As in \cite{crawford2019predicting}, when analyzing WECTs, we often preprocess them to improve robustness, by applying a smoothing operator. Unlike \cite{crawford2019predicting}, we do not specify a particular smoothing operation, and leave the particular method as a hyperparameter in the data analysis pipeline.

\begin{figure}[!t]
    \centering
    \includegraphics[width = 0.4\textwidth]{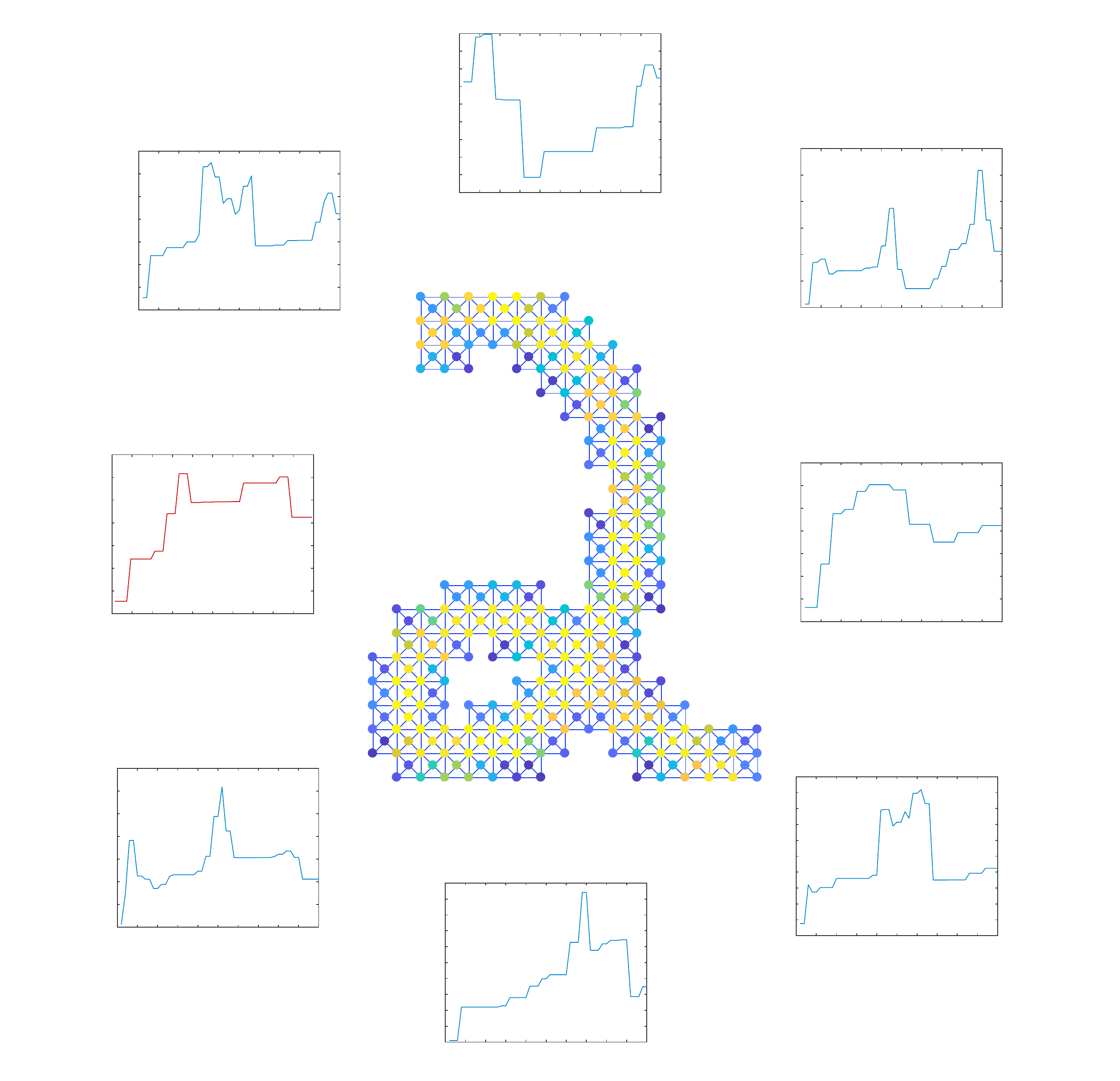}
    \caption{The WECT for a weighted simplicial complex constructed from an MNIST digit. Each panel shows a single weighed Euler curve, with the red curve on the left representing filtering by projection to the vector $(-1,0)$, and the other curves constructed similarly by projection onto other directions.}
    \label{fig:MNIST_WECT}
\end{figure}

\subsection{Distance Between WECTs}

The WECT of a weighted simplicial complex $(K,g)$ in $\R^d$ is naturally viewed as a family of integer-valued functions $\wect_{K,g}(v,\cdot):\R \rightarrow \Z$, parameterized by $S^{d-1}$. Since $K$ is assumed to be compact, each function is constant outside of a compact subset of $\R$, and we may restrict each function to this common compact domain; moreover, given a dataset of weighted simplicial complexes, one may assume without loss of generality that all WECT functions are defined on the same compact domain. After applying a smoothing operator, the smoothed WECT is likewise identified with a parameterized family of compactly supported functions of higher regularity. Any metric $d$ on such functional data gives rise to a metric on WECT data, by integrating the function 
$$
v \mapsto d(\wect_{K_1,g_1}(v,\cdot),\wect_{K_2,g_2}(v,\cdot))
$$
over $v \in S^{d-1}$ with respect to its standard volume form.

The most convenient metric on compactly supported functions is the one induced by the standard $L^2$ norm (with respect to Lebesgue measure), denoted $\|\cdot\|_{L^2}$. We abuse notation slightly and denote the induced metric on the space of WECTs also using norm notation as follows:
\begin{equation}\label{eqn:wect_distance}
\| \wect_{K_1,g_1} - \wect_{K_2,g_2} \|_{L^2}.
\end{equation}
This notation is in fact warranted, since this metric is equivalent to the one induced by the $L^2$ norm on $S^{d-1} \times I$, where $I$ is a compact interval, with respect to the product of the standard measure on $S^{d-1}$ with Lebesgue measure on $I$. With this metric, the space of WECTs has a Euclidean structure, meaning that WECTs are amenable to methods from functional data analysis and machine learning.

Computationally, a WECT is represented by a finite number of samples. Taking $m$ samples from $\R$ and $n$ samples from $S^{d-1}$, the values of the WECT can be arranged in a matrix of size $m \times n$. Then, the $L^2$ distance in Equation \eqref{eqn:wect_distance} can be computed simply as a Frobenius norm, making the process of comparing WECTs numerically efficient.

\subsection{Injectivity of the WECT}

Inverse problems in topological data analysis have recently become an active topic of research \cite{oudot2019inverse}. The basic general question is: Is it possible for inequivalent spaces to be mapped to the same topological summary statistic? This question has recently been tackled for various flavors of topological signatures \cite{frosini2011uniqueness,oudot2017barcode,curry2018fiber,catanzaro2019moduli} including Persistent Homology and Euler Curve Transforms \cite{turner2014persistent,ghrist2018persistent,fasy2018challenges,curry2018many,betthauser2018topological,crawford2019predicting,fasy2019persistence}.  

The original paper on the ECT \cite{turner2014persistent} demonstrated a uniqueness result for ECT representations of compact embedded simplicial complexes with an algorithmic proof. This perspective has been pushed further to provide a sufficient number of direction samples to guarantee injectivity \cite{curry2018many}. It is shown in \cite{betthauser2018topological} that for weighted \emph{cubical complexes} defined on a regular axis-aligned lattice in $\R^d$, only $2^d$ generic samples are sufficient and an explicit reconstruction algorithm is provided. Our Algorithm \ref{alg:greyscale} produces a simplicial complexes which is essentially equivalent to the cubical complexes of \cite{betthauser2018topological}, so the reconstruction results their can be ported over directly to weighted simplicial complexes constructed via Algorithm \ref{alg:logmap}. 

In anticipation of the possibility of studying non-axis-aligned weighted simplicial complexes through the WECT signature, one might hope for a more general injectivity result. An alternative approach to the injectivity question for ECTs is given in \cite{ghrist2018persistent,curry2018many}. In these articles, the theory of Euler integral calculus is employed to prove injectivity. This approach is more theoretical and comes with the cost of a less explicit inversion algorithm. This is balanced by more general applicability. In particular, one has the following, quite general, result.

\begin{thm}[Theorem 1, \cite{ghrist2018persistent}]\label{thm:ghrist}
The map 
$$
\mathcal{R}:\mathrm{CF}_c(\R^d) \rightarrow \mathrm{CF}(S^{d-1} \times \R)
$$
defined by
\begin{equation}\label{eqn:radon_transform}
\left(\mathcal{R}(g)\right)(v,r) = \int_{\R^d} g(x) \cdot \textbf{1}_{x \cdot v \leq r} \; d\chi(x)
\end{equation}
is injective.
\end{thm}

We use $\mathrm{CF}(\R^d)$ to denote the space of \emph{constructible functions}; these are functions $\R^d \rightarrow \Z$ whose level sets satisfy a certain tameness condition, defined nowadays in the technical language of $o$-minimal set theory \cite{baryshnikov2011inversion,curry2012euler,ghrist2018persistent}. The set $\mathrm{CF}(S^{d-1} \times \R)$ is defined similarly. We are restricting to \emph{compactly supported} constructible functions $\mathrm{CF}_c(\R^d)$. This space in particular contains admissible functions defined on embedded simplicial complexes in $\R^d$. The right side of Equation \eqref{eqn:radon_transform} is defined in terms of Euler integration. Roughly, one treats the Euler characteristic formally as a measure, allowing for integration of sufficiently well-behaved functions. The transform $\mathcal{R}$ can be understood as a topological version of the classical Radon transform used in tomography applications \cite{helgason1980radon}. Theorem \ref{thm:ghrist} is proved by appealing to a general result of Schapira on inverting topological Radon transforms of this type \cite{schapira1995tomography}. 
The authors of \cite{ghrist2018persistent} observe that if $g$ is the indicator function for an embedded simplicial complex $K$, then $\mathcal{R}(g)$ is exactly the ECT for $K$, whence the ECT is injective \cite[Corollary 1]{ghrist2018persistent}. On the other hand, if we consider functions $g$ which are admissible weight functions on embedded simplicial complexes, we obtain the following result as an immediate corollary.

\begin{thm}\label{thm:injectivity_theorem}
The Weighted Euler Characteristic Transform is injective on the space of weighted simplicial complexes. That is, if $(K_1,g_1)$ and $(K_2,g_2)$ are weighted simplicial complexes in $\R^d$ with $\wect_{K_1,g_1} = \wect_{K_2,g_2}$, then $(K_1,g_1)=(K_2,g_2)$. 
\end{thm}

\subsection{Comparison to Other Methods}\label{sec:comparison_to_other_methods}

The WECT provides a topological signature which simultaneously incorporates shape data and non-geometric weight data. In the case of image data, by discretely sampling the domain $S^{d-1} \times \R$ one obtains a discrete signature with a similar memory footprint to the original image. However, we show experimentally that the WECT provides a representation, which is more effective at distinguishing shape features. In this subsection, we compare the WECT representation to other shape descriptors appearing in the topological data analysis literature.

\noindent\textbf{Persistent Homology.} The WECT representation has several benefits over the commonly used persistence diagram signature. Foremost, it is a nontrivial task to simultaneously incorporate geometric and non-geometric features into a persistence diagram. One approach is to use a multiparameter filtration of the dataset \cite{frosini1999size,carlsson2009theory}. The major drawback of such an approach is that multiparameter persistent homology does not in general admit a convenient analogue of the persistence diagram statistics used in classical persistent homology. An alternative approach to incorporating geometric and non-geometric features into persistent cohomology was recently proposed in \cite{cang2018persistent}, where an enriched barcode representation is  obtained through least squares optimization of persistent cohomology cycle representatives.

The simple WECT representation for weighted simplicial complex data also has the benefit of immediately providing a vectorized topological signature. This allows straightforward usage of WECT summaries as covariates in statistical models---this was the main idea of \cite{crawford2019predicting}, where the ECT summaries were used as covariates in a Gaussian process regression for prediction of survival times of subjects with Glioblastoma Multiforme brain tumors. This is in stark contrast to analysis using persistence diagrams or barcodes from persistent homology. Indeed, a persistence diagram is an unstructured point cloud in $\R^2$ and care must be taken to vectorize this signature in order to incorporate it into statistical models. There are several extant vectorization methods in the literature, including \emph{persistence landscapes} \cite{bubenik2015statistical} and \emph{persistence images} \cite{adams2017persistence}, as well as more straightforward feature aggregation \cite{bendich2016persistent}. Any vectorization of the persistence diagram space necessarily distorts its natural latent geometry, since the canonical metric on persistence diagrams, the bottleneck distance, is non-Euclidean \cite{bubenik2019embeddings}. 

\noindent\textbf{Variants of the ECT.} When studying simplicial complexes arising from grayscale image data, one could imagine other relevant simplicial complexes to which one could apply the standard ECT. Examples include thresholding pixel values in the image and building restricted two-dimensional complexes or using the pixel values to build a three-dimensional simplicial complex. We found these approaches to give unsatisfactory performance on our tumor dataset, although they may be viable approaches for other applications.

\section{Applications}\label{sec:applications}

\subsection{Classification of MNIST Digit Images}

To understand the descriptive power of the WECT representation of image data, we first explore its ability to classify images from the ubiquitous MNIST handwritten digit dataset \cite{lecun1998gradient}. We use a small subset of 1000 $28 \times 28$ grayscale images, evenly distributed over 10 digits $0,1,\ldots,9$. As a baseline, we treat each image as a vector in $\R^{28 \times 28}$ and classify them using Support Vector Machines (SVM) with a linear kernel. Next, we produce WECT representations of all digit images. In this experiment, we discretize $S^1 \times \R$ into a $25 \times 50$ grid (i.e., 25 Euler curve directions, 50 points along each curve domain). We also smooth the Euler curves to improve robustness using a Gaussian kernel with window size $0.2 \cdot 50$ (these particular parameters were chosen in a tuning step, but we found that the results are generally insensitive to the parameter choice). We then considered each WECT representation as a vector in $\R^{25 \times 50}$ and classified using SVM with a linear kernel. We also produced smoothed ECT representations with similar parameters and ran an SVM classification. The ten-fold cross-validated classification rates from these experiments are displayed in Table \ref{tab:exp1}.

\begin{table}[!t]
\caption{SVM ten-fold classification performance of vectorized image, ECT and WECT representations for the MNIST digit data.}
\label{tab:exp1}
\begin{center}
\begin{tabular}{|c|c|}
\hline
Representation & Classification Rate \\
\hline
Image $\R^{28 \times 28}$  & 87.84 $\pm$ 1.42 \% \\
\hline
ECT $\R^{25 \times 50}$ & 89.88 $\pm$ 1.66 \% \\
\hline
WECT $\R^{25 \times 50}$ &  94.68 $\pm$ 1.57 \% \\
\hline
\end{tabular}
\end{center}
\end{table}

The classification results show that the WECT representation of the digit images is adept at encoding and distinguishing shape features, while having a similar memory footprint to the original image representation. It also outperforms the classification using smoothed ECT representations. We stress that this classification result is, of course, not meant to be competitive with those obtained by deep learning methods. Rather, this simple experiment suggests that the WECT representation produces an interesting shape summary for this type of image data, which is computationally efficient and can be trivially incorporated into various statistical models. 

To get a more detailed qualitative picture of the differences between the raw image, ECT and WECT representations of the MNIST image data, we also computed t-SNE embeddings  \cite{maaten2008visualizing} for each representation; see Figure \ref{fig:tSNE_embeddings}. While class separation is apparent in all three embeddings, it is immediately evident that the embeddings of the ECTs and WECTs are  much more distinctly clustered. On the other hand, one can easily see how classification errors arise in the ECT embedding. We believe that these errors occur because the ECT is more sensitive to topological differences between digits, while the WECT smooths these differences using weight data.

\begin{figure}[!t]
	\centering
	\includegraphics[width = 0.15\textwidth]{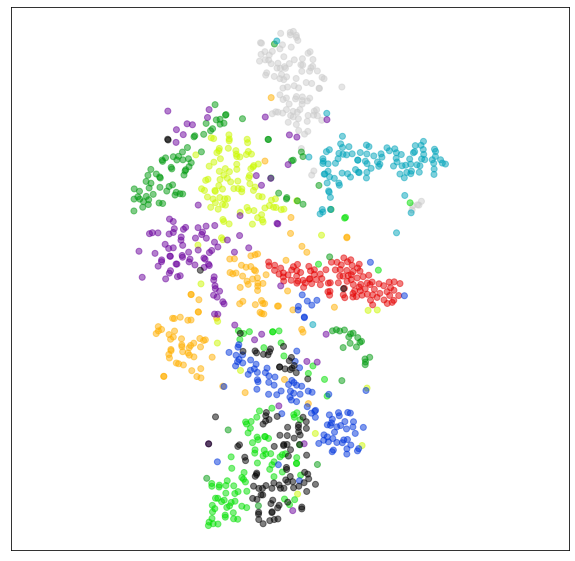}
	\includegraphics[width = 0.15\textwidth]{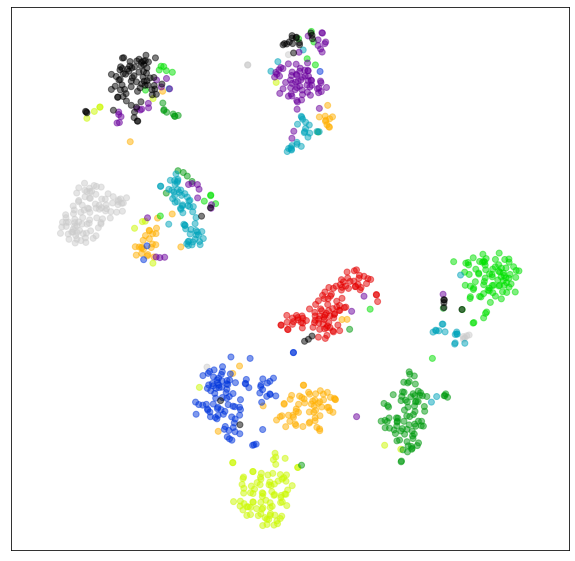}
	\includegraphics[width =
	0.15\textwidth]{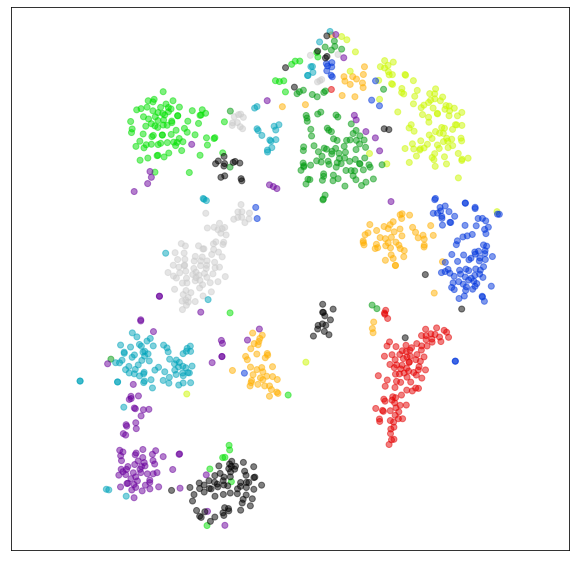}
	\caption{T-SNE embeddings of the MNIST image dataset. Left: Raw image vectors. Middle: Smoothed ECTs. Right: Smoothed WECTs.}
	\label{fig:tSNE_embeddings}
\end{figure}

\subsection{Rigid and Scale Registration}\label{sec:scale_rotation}

\begin{figure}[!t]
    \centering
    \includegraphics[width = 0.09\textwidth]{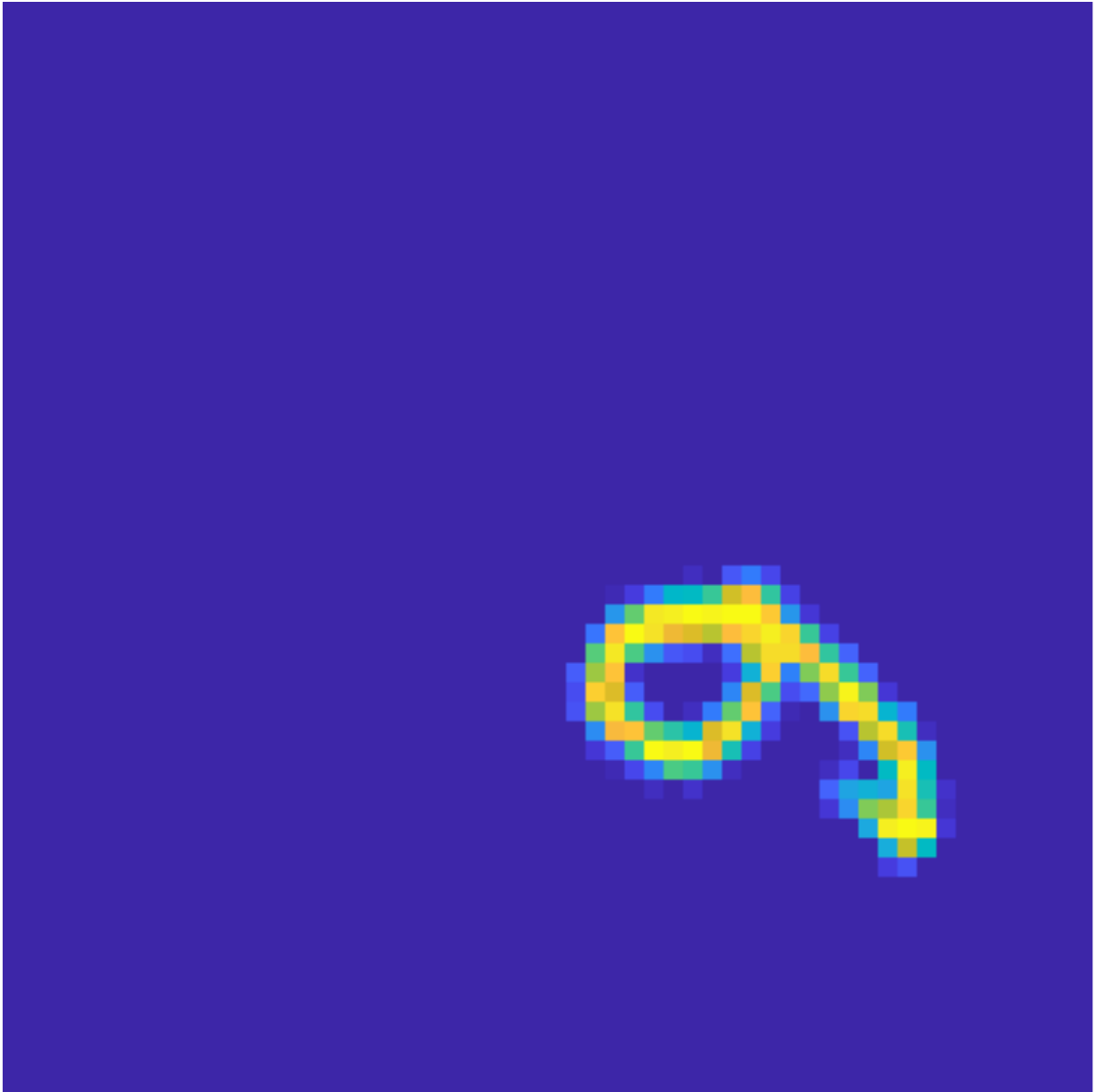}
    \includegraphics[width = 0.09\textwidth]{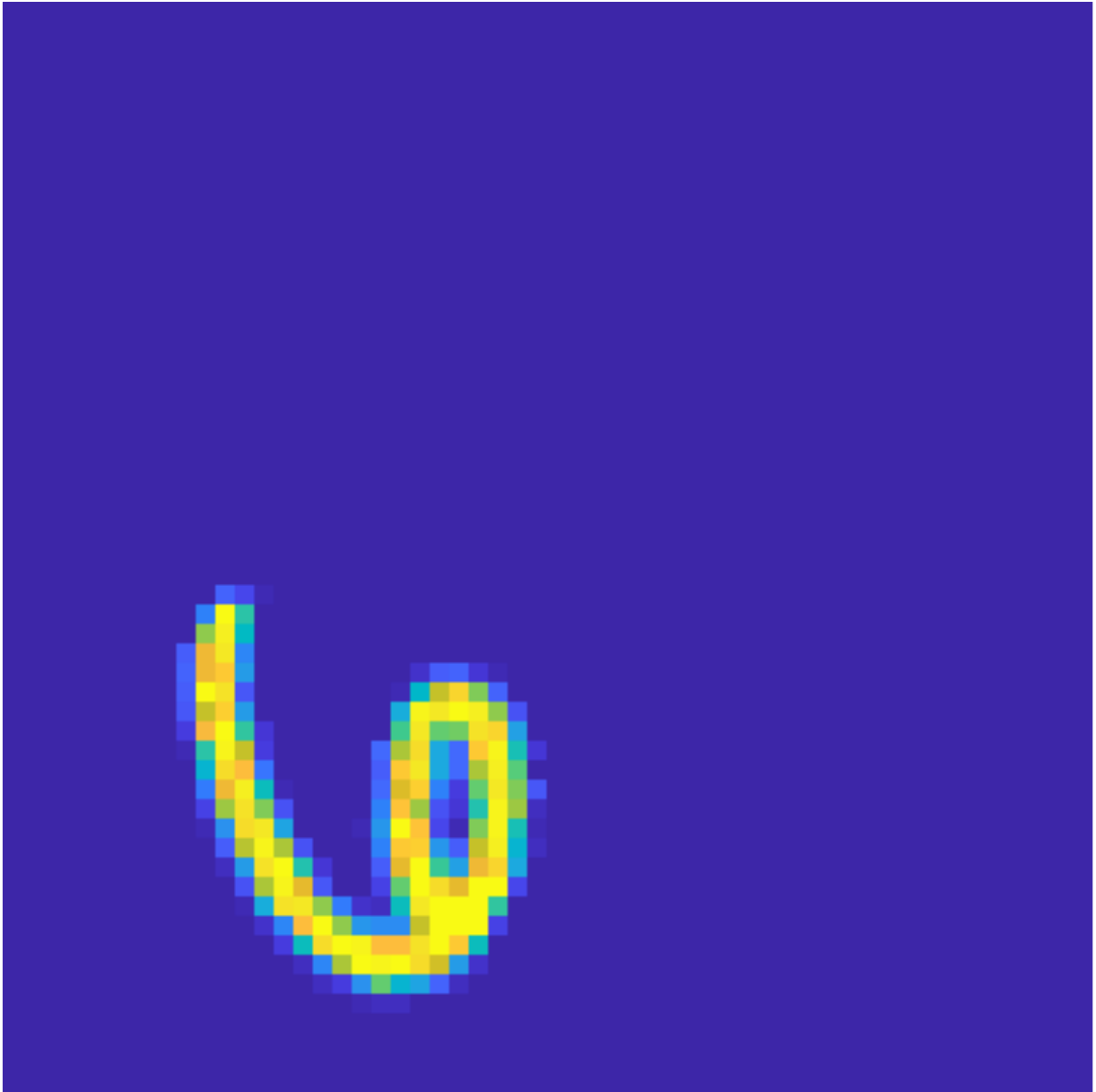}
    \includegraphics[width = 0.09\textwidth]{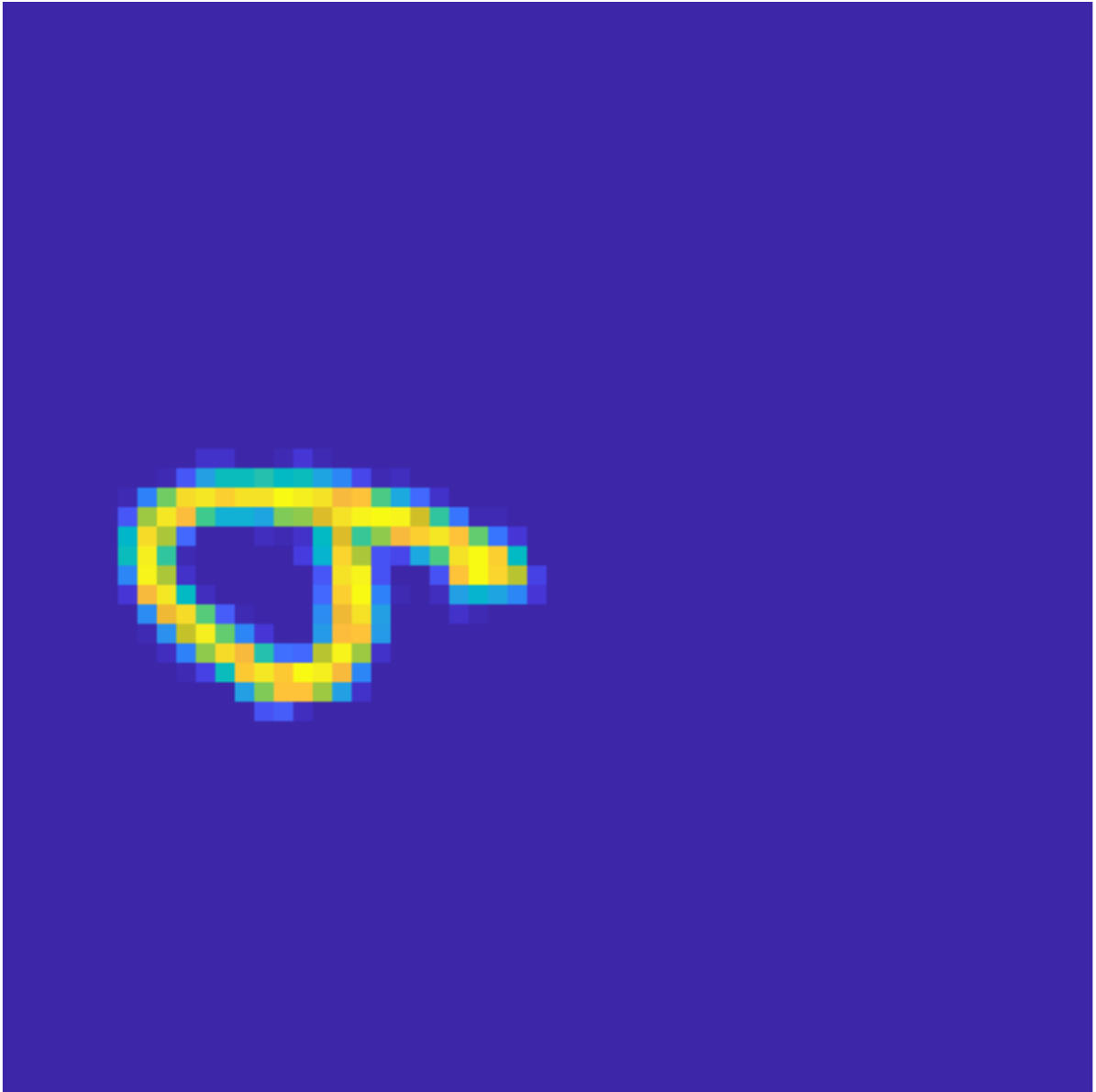}
    \includegraphics[width = 0.09\textwidth]{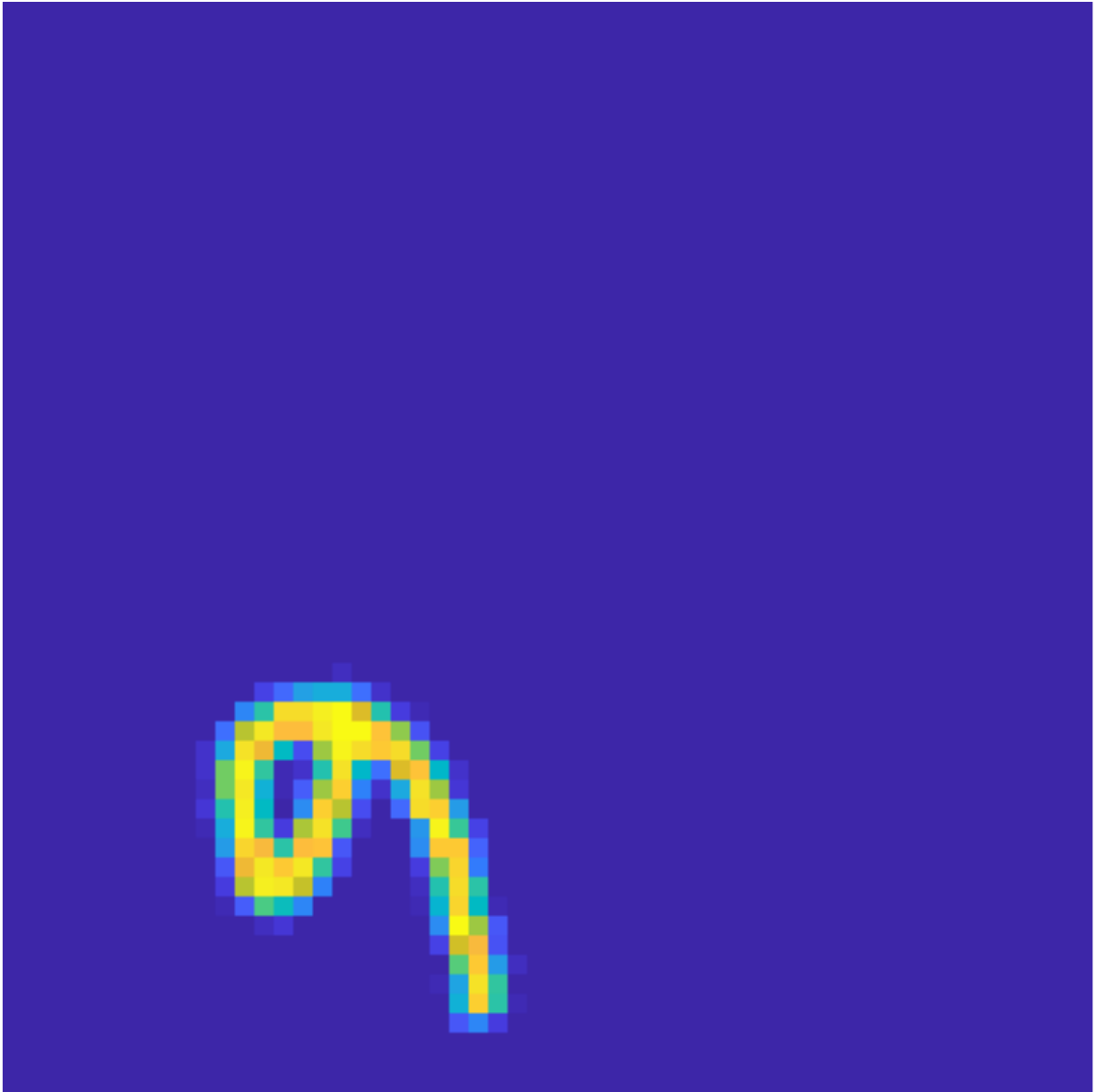}
    \includegraphics[width = 0.09\textwidth]{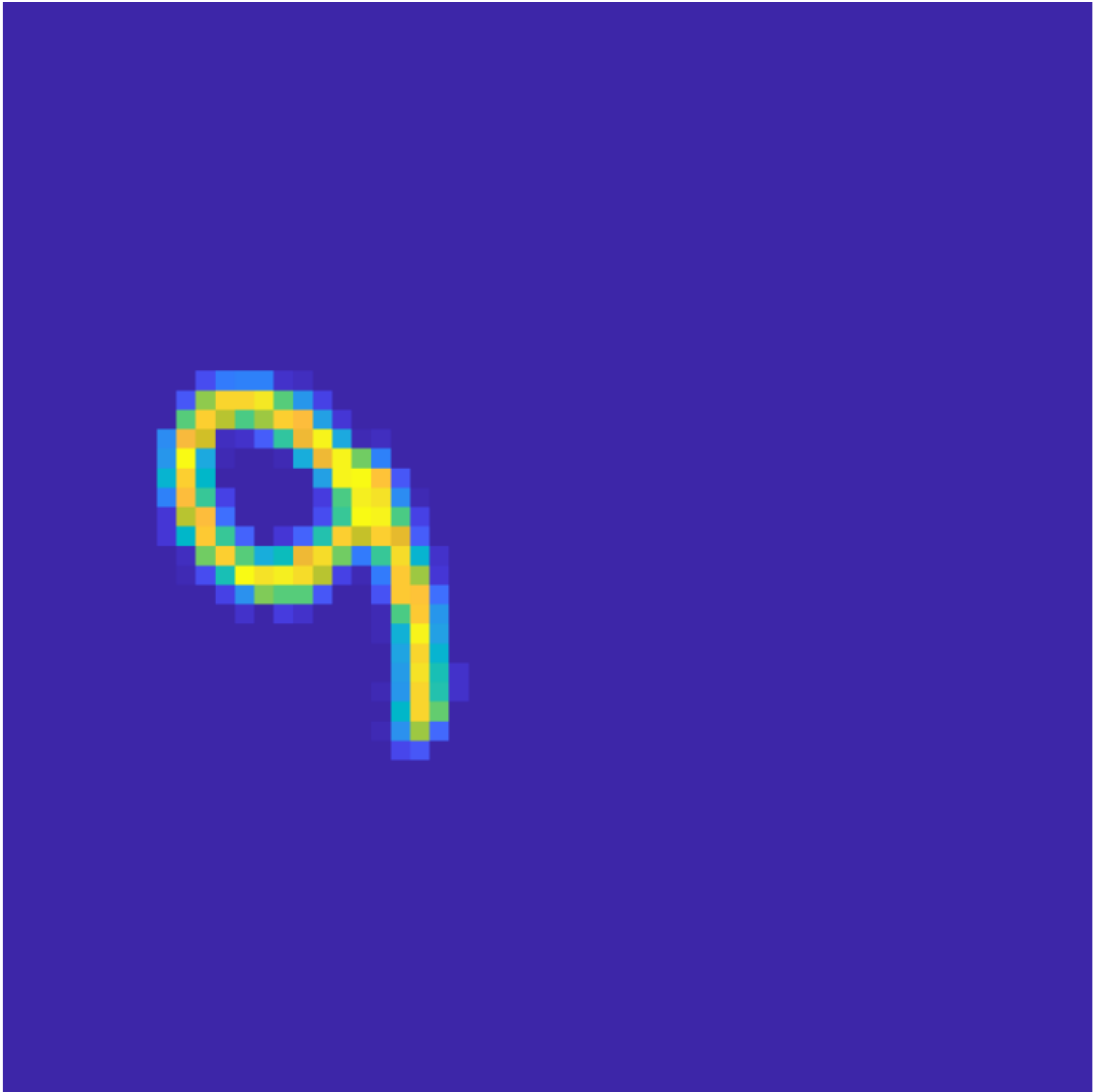}
    \includegraphics[width = 0.09\textwidth]{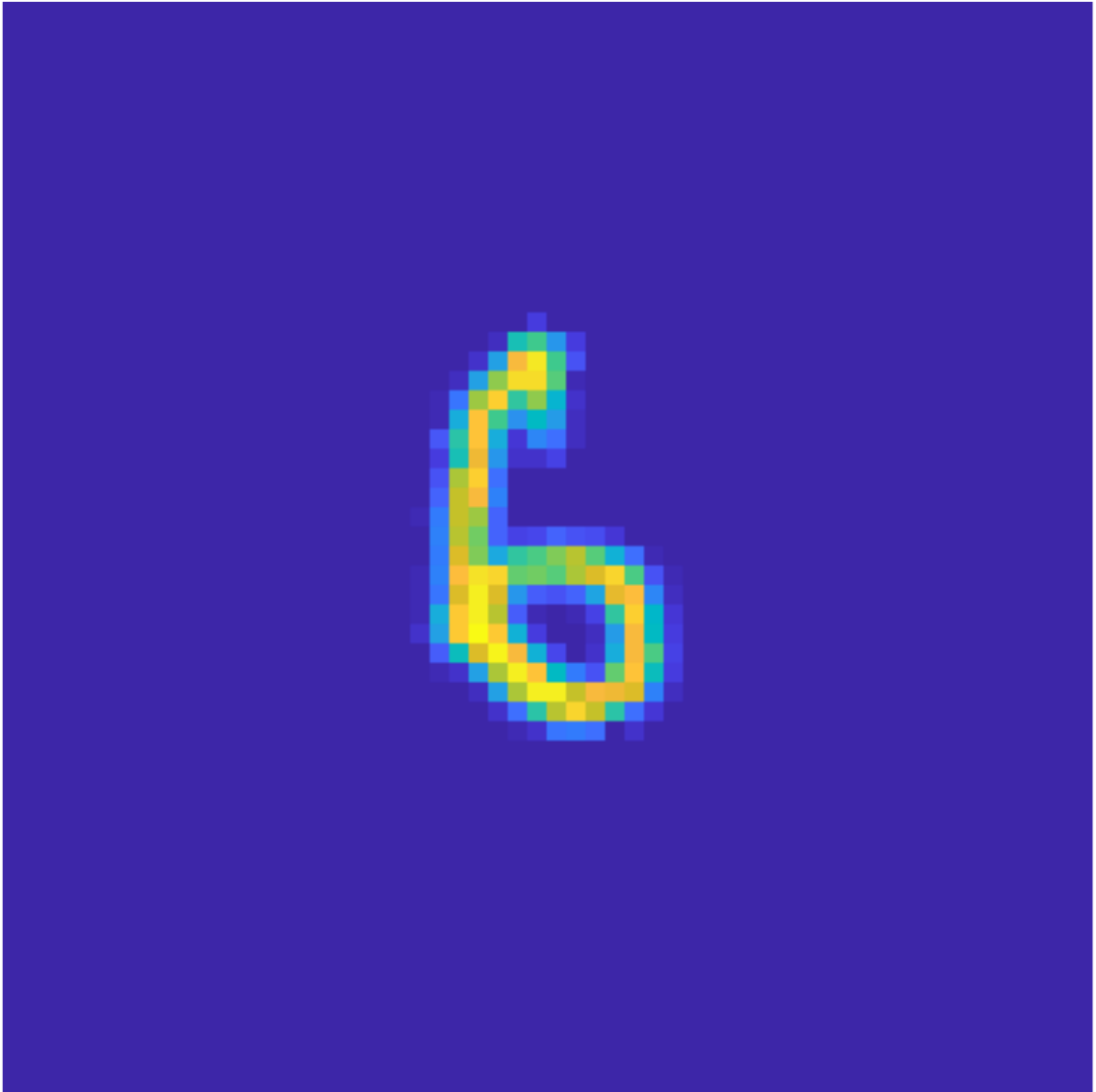}
    \includegraphics[width = 0.09\textwidth]{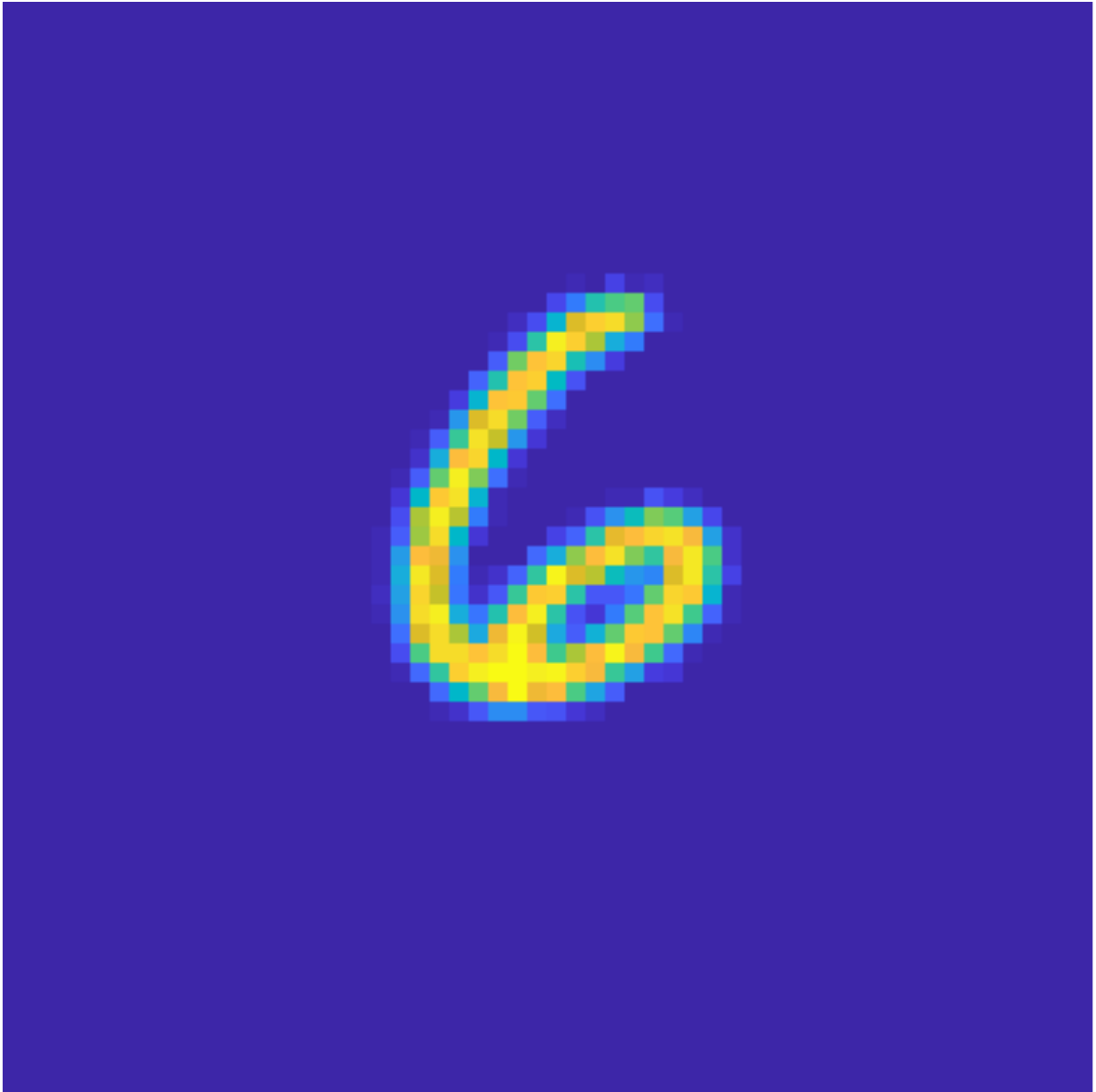}
    \includegraphics[width = 0.09\textwidth]{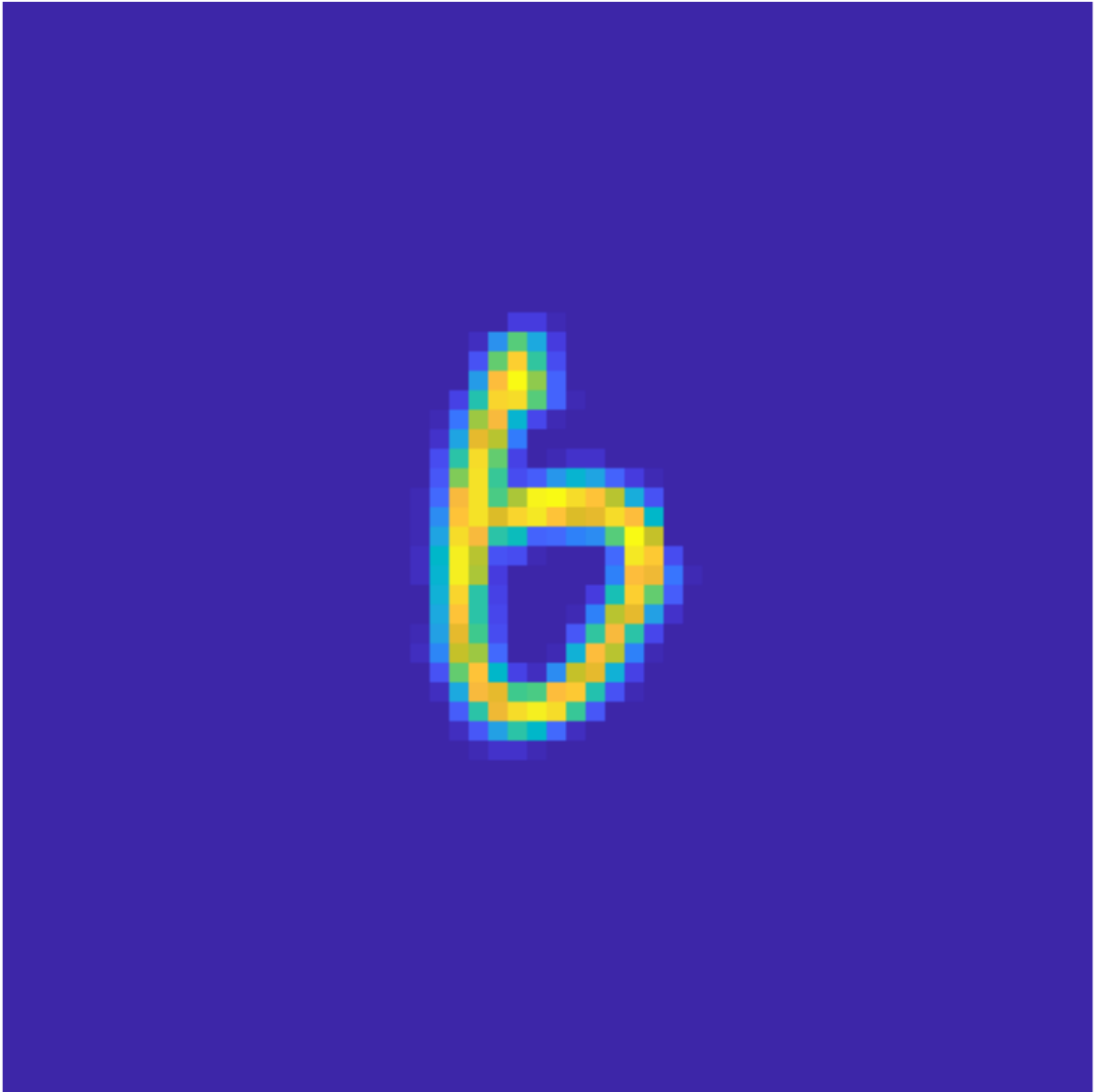}
    \includegraphics[width = 0.09\textwidth]{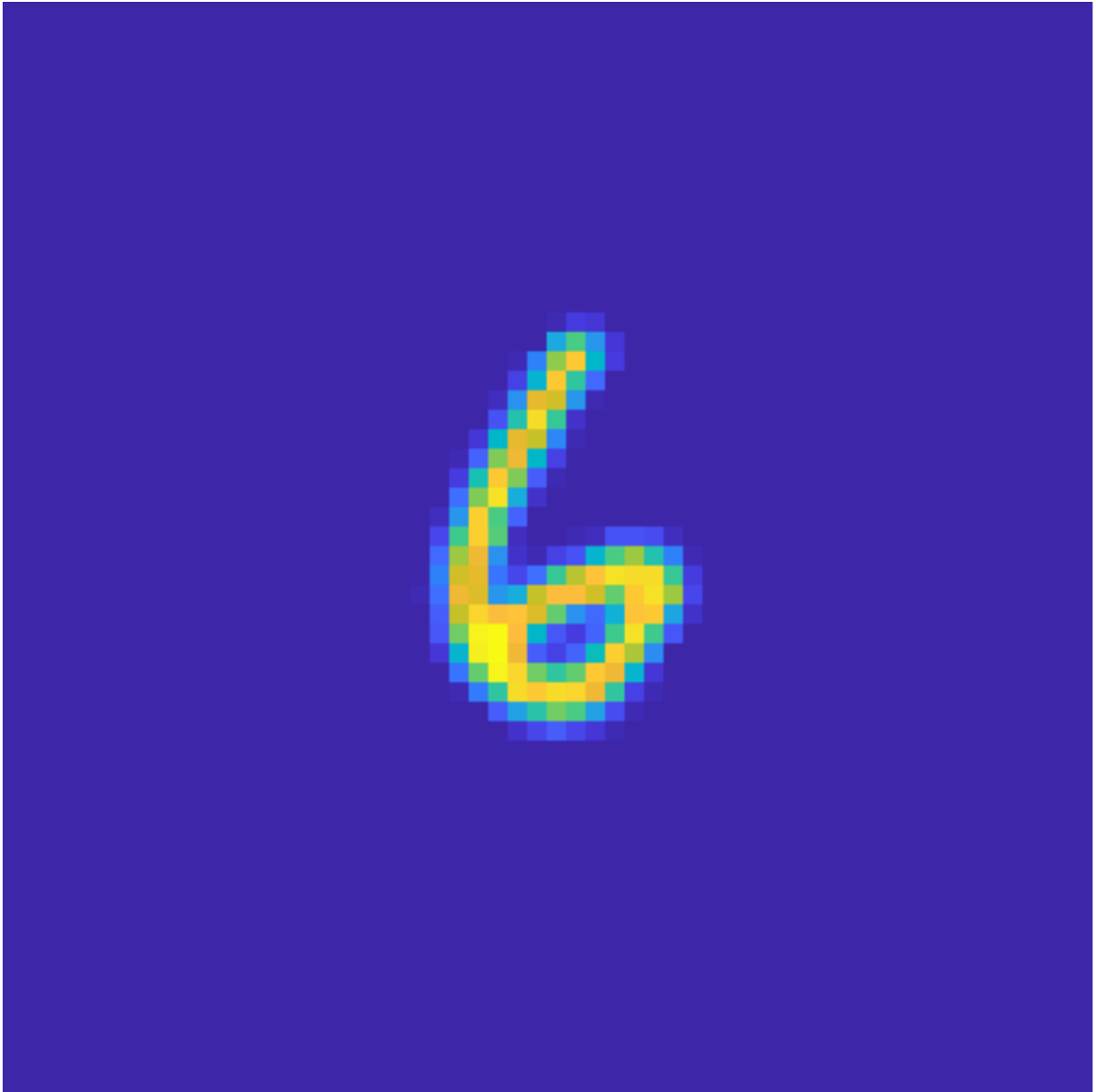}
    \includegraphics[width = 0.09\textwidth]{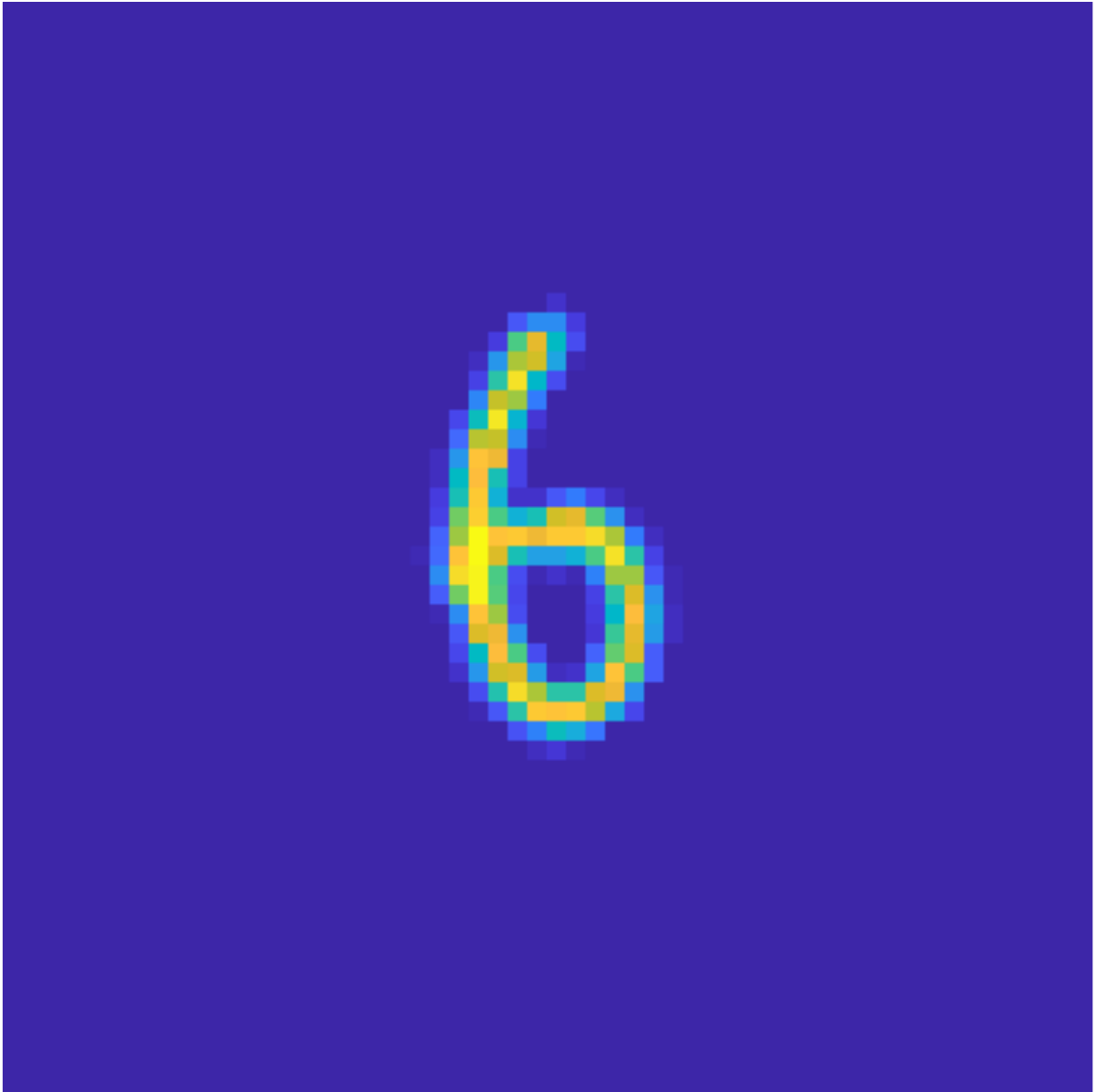}
    \caption{Top: MNIST digit images randomly rotated and translated. Bottom: The same digits after rigid registration to a template digit via the process described in Section \ref{sec:scale_rotation}.}
    \label{fig:MNIST_registration}
\end{figure}

One benefit of the simplicial complex representation of image data is that registering over scale and rigid transformations (translations and rotations) becomes trivial. Once a pair of images have been converted to weighted simplicial complexes $(K_1,g_1)$ and $(K_2,g_2)$, they can be immediately registered with respect to translation and scaling by centering each complex at the origin, and normalizing (treating vertex locations as vectors). To register over rotations, one then computes weighted Euler characteristic transforms $\wect_{K_j,g_j},\ j=1,2$ and solves the optimization problem
\begin{equation}\label{eqn:optimization}
\min_{R \in \mathrm{SO}(2)} \|\wect_{K_1,g_1} - R \cdot \wect_{K_2,g_2}\|_{L^2},
\end{equation}
where the rotation group $\mathrm{SO}(2)$ acts on a WECT by precomposition in the $S^1$-coordinate. As was noted above, the $L^2$ distance is numerically trivial to compute for finite approximations of WECTs. Thus, the optimization problem in Equation \eqref{eqn:optimization} can be solved quickly by an exhaustive search over cyclic permutations of the WECT matrix. The minimizing rotation $
R$ can then be used to register $(K_2,g_2)$ to $(K_1,g_1)$ with respect to rotations---see Figure \ref{fig:MNIST_registration}.

\subsection{Analysis of GBM Tumor Data}

Glioblastoma Multiforme (GBM) is the most common malignant brain tumor in adults \cite{holland2000}; for most patients, the prognosis is very poor: less than $10\%$ of individuals survive longer than five years and the median survival time is approximately 12 months \cite{tutt2011,mcnamara2013,mclendon2008}. GBM is a morphologically heterogeneous disease. GBM tumors exhibit complex structure in terms of their overall shape as well as internal makeup. Often, dead cells are present inside the tumor and increased blood flow near the boundary of the tumor \cite{marusyk2012}. These features result in various pixel value patterns of GBM tumor images. Thus, characterization of both the shape and texture of GBM tumors, based on medical imaging data, is important for disease prognosis as well as survival prediction. While previous studies have considered these two features separately \cite{bharath2018radiologic,Saha2016132} in the analysis, our approach is to analyze them jointly under a unified representation.

In this study, we use T1-weighted post contrast magnetic resonance images (MRIs) of GBM tumors from 63 subjects. For our analysis, we select a single axial slice with largest tumor area from each 3D image (the same approach was taken in \cite{bharath2018radiologic,Saha2016132}), and summarize the tumors' shapes and textures via the WECT. For details on the image pre-processing steps that were used prior to our analysis, see \cite{Saha2016132}.

We use a simple distance-based clustering approach to analyze the tumor data. First, each of the 63 tumor images is converted into a weighted simplicial complex using Algorithm \ref{alg:greyscale}. To isolate the shape and weight information, all simplicial complexes are centered at the origin and normalized so that the vertex farthest from the origin is at distance $1$. The weights of the simplicial complexes are then normalized to have maximum weight one; this was done to account for the varying pixel value distributions of the MRIs for each subject. Next, each weighted simplicial complex is given a smoothed WECT representation. Specifically, for each tumor image, we use 25 directions and 50 points along the domain of the Euler curve for each direction. The Euler curves were smoothed using a Gaussian kernel with a smoothing window of ten. Next, the $L^2$ distance between each pair of smoothed WECT representations was computed with registration of the tumor images over rotations (see Section \ref{sec:scale_rotation}). We applied hierarchical clustering with Ward linkage to the $63 \times 63$ distance matrix, which first suggested three natural clusters. The clusterwise mean and median survival times (in months) are reported in Table \ref{tab:exp2}.

\begin{table}[!t]
\caption{Clusterwise mean and median survival.}
\label{tab:exp2}
\medskip
\begin{center}
\begin{small}
\begin{tabular}{|c|c|c|c|}
\hline
Mean & 6.7 & 12.9 & 20.2  \\
\hline
Med. & 6.2 & 9.6 & 15.2  \\
\hline
\end{tabular}
\end{small}
\end{center}
\end{table}

\begin{figure}[!t]
    \centering
    \includegraphics[width = 0.09\textwidth]{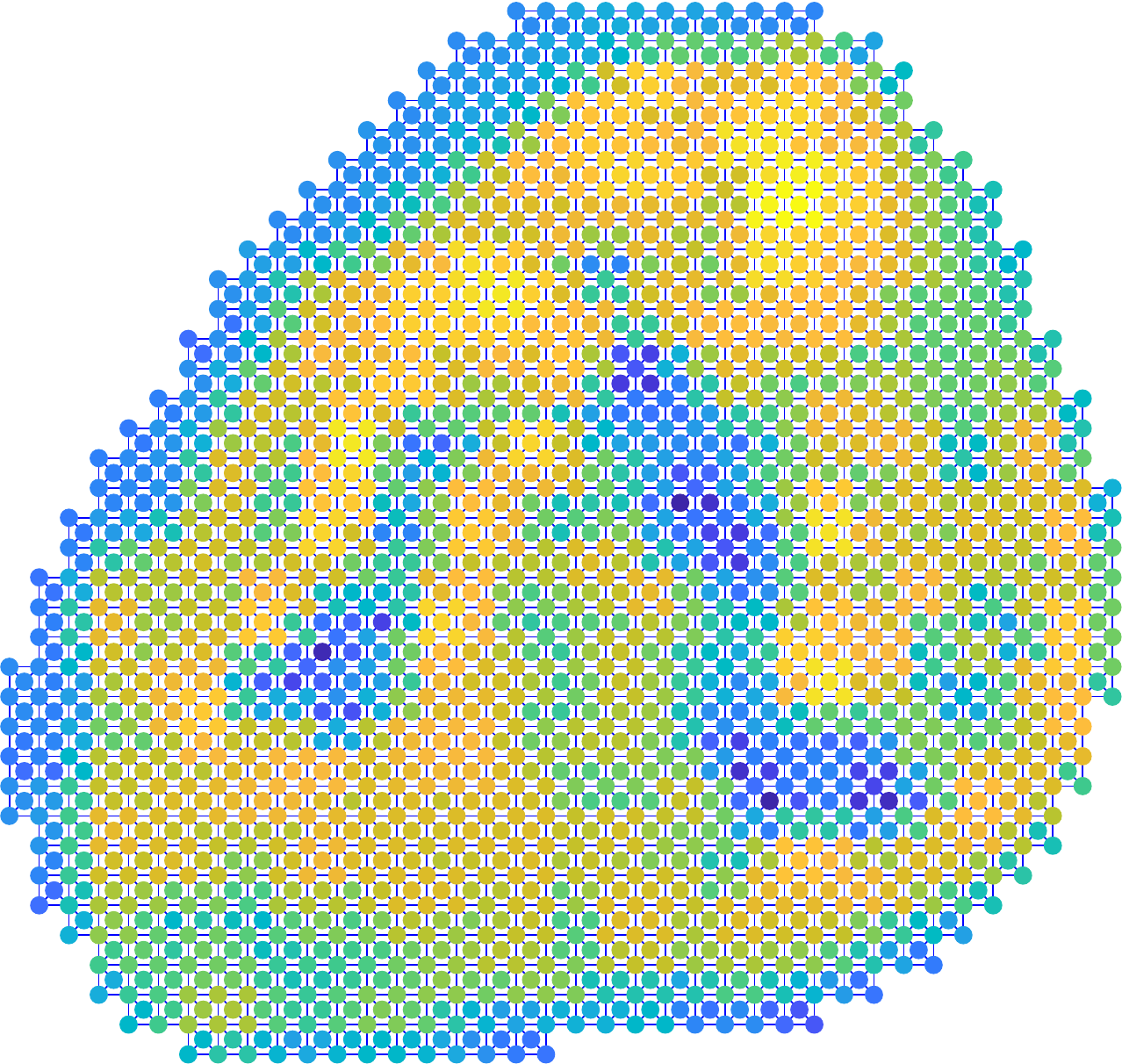}
    \includegraphics[width = 0.09\textwidth]{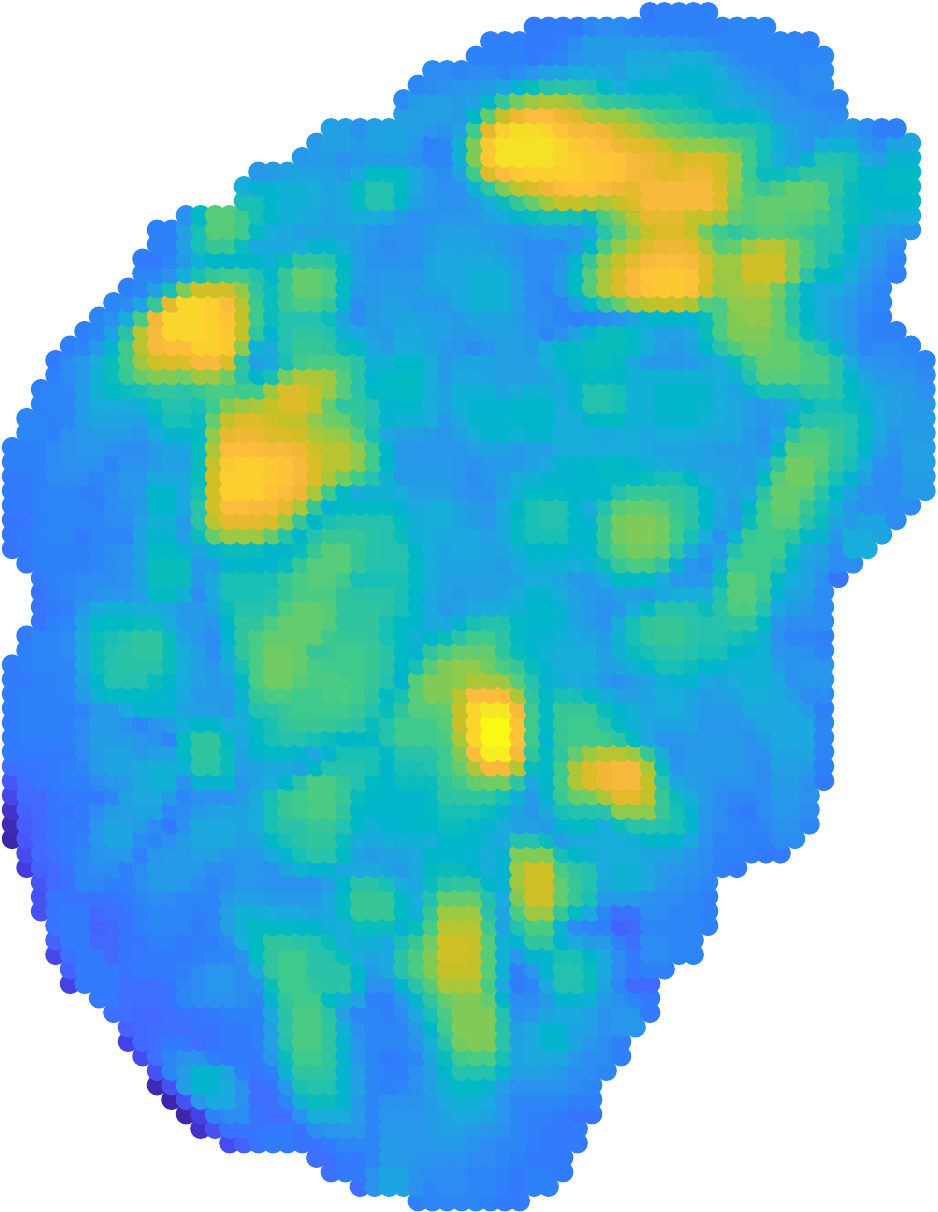}
    \includegraphics[width = 0.09\textwidth]{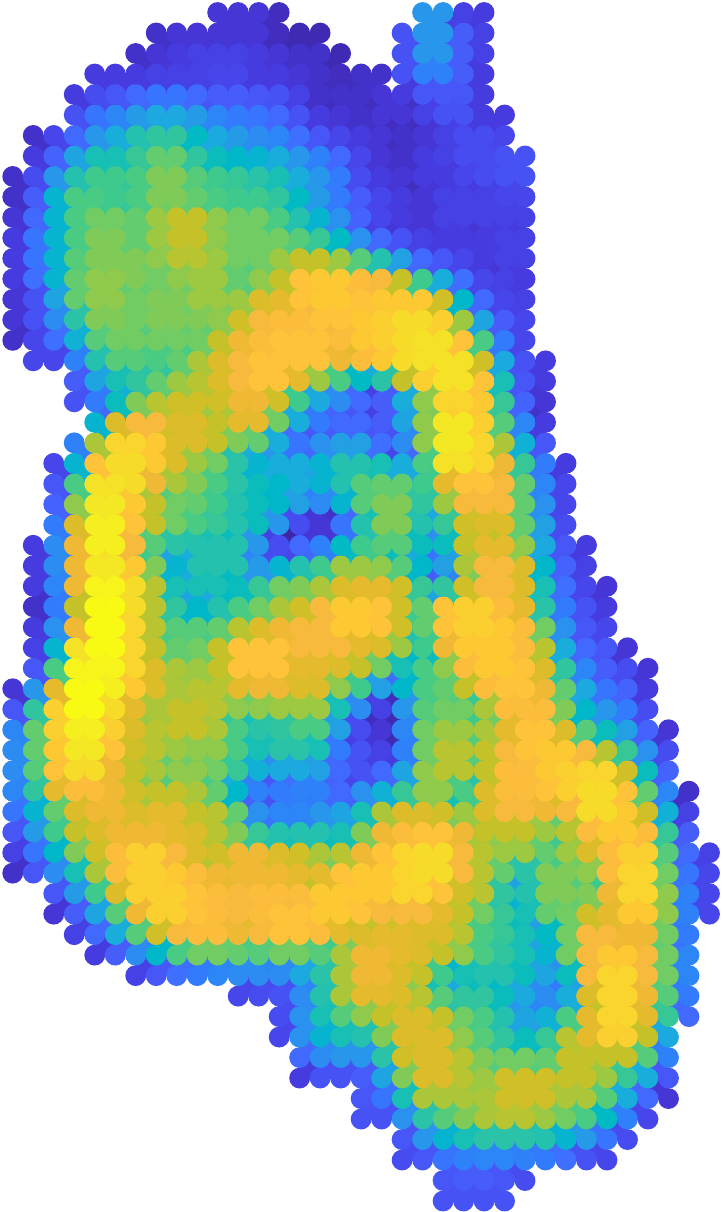}
    \includegraphics[width = 0.09\textwidth]{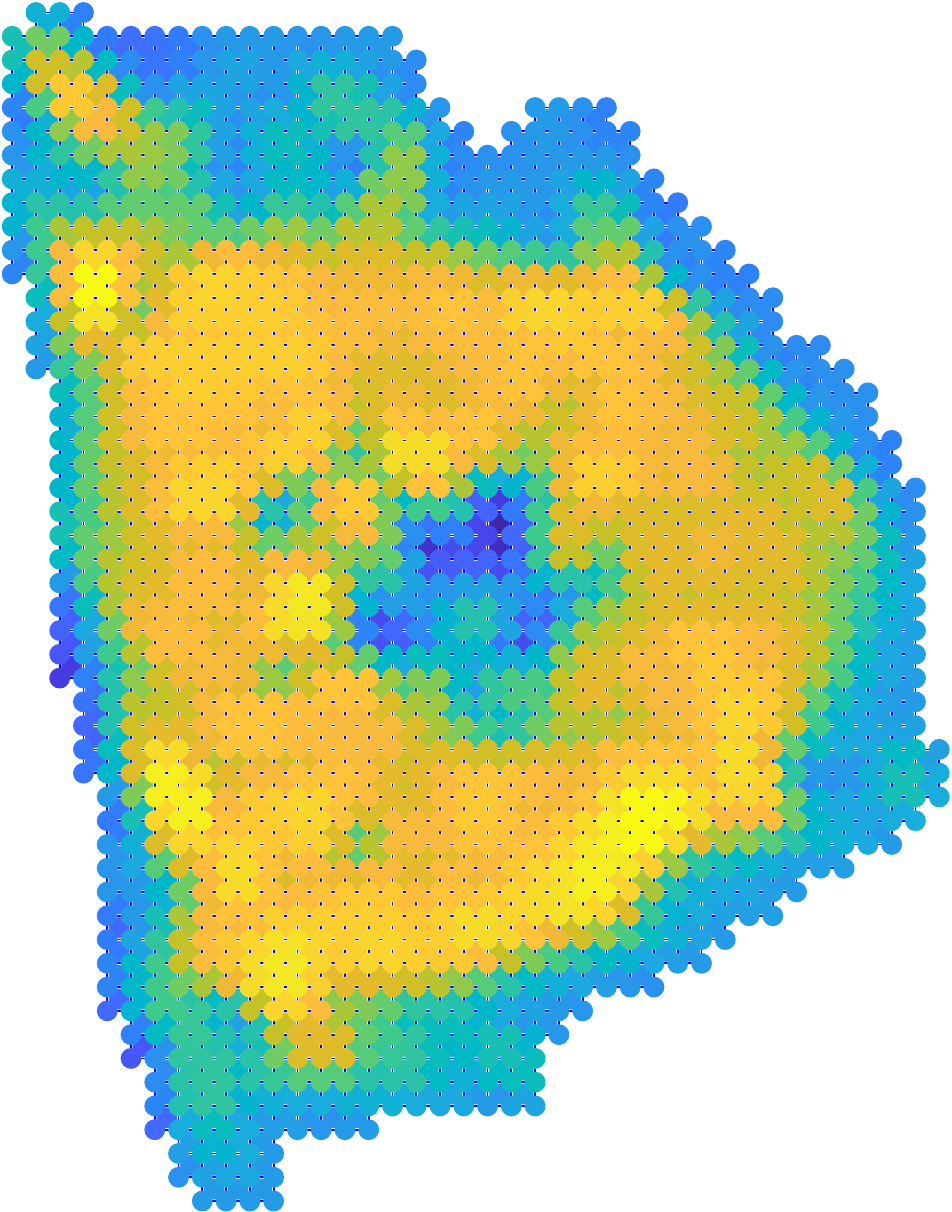}
    \includegraphics[width = 0.09\textwidth]{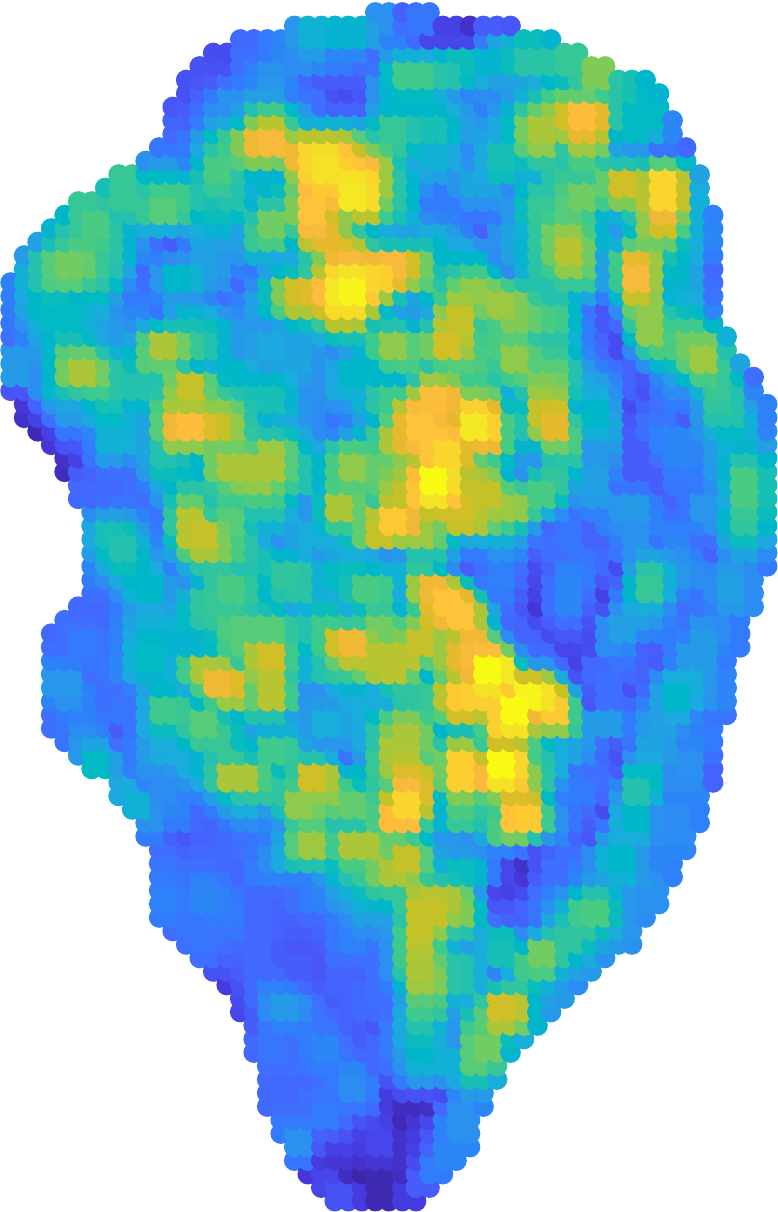}
    \caption{Weighted simplicial complex representations of tumors from the low survival time cluster in Table \ref{tab:exp2}.}
    \label{fig:tumor_cluster_outliers}
\end{figure}

These statistics suggest that the clusters are roughly characterized as low, medium and high survival. Figure \ref{fig:tumor_cluster_outliers} shows tumors from the low survival cluster; they are visually irregular in shape and intensity distribution, which explains their presence as a distinct cluster. To explore the data in more depth, we consider the clustering dendrogram with this cluster of tumors removed. Figure \ref{fig:tumor_dendrogram_6} shows this dendrogram on the remaining 58 tumors, with six highlighted clusters; mean and median survival times for patients in these clusters are shown in Table \ref{tab:exp3}. Inspecting the tumors in these clusters, one can observe various common qualitative shape and intensity features. For example, the tumors in the blue and cyan clusters both tend to have intensity patterns with a ring-like structure near the boundary. The tumors in the blue cluster tend to have higher irregularity in shape and/or intensity patterns, see Figure \ref{fig:tumor_cluster_cyan_and_blue}.

\begin{figure}[!t]
    \centering
    \includegraphics[width = 0.4\textwidth]{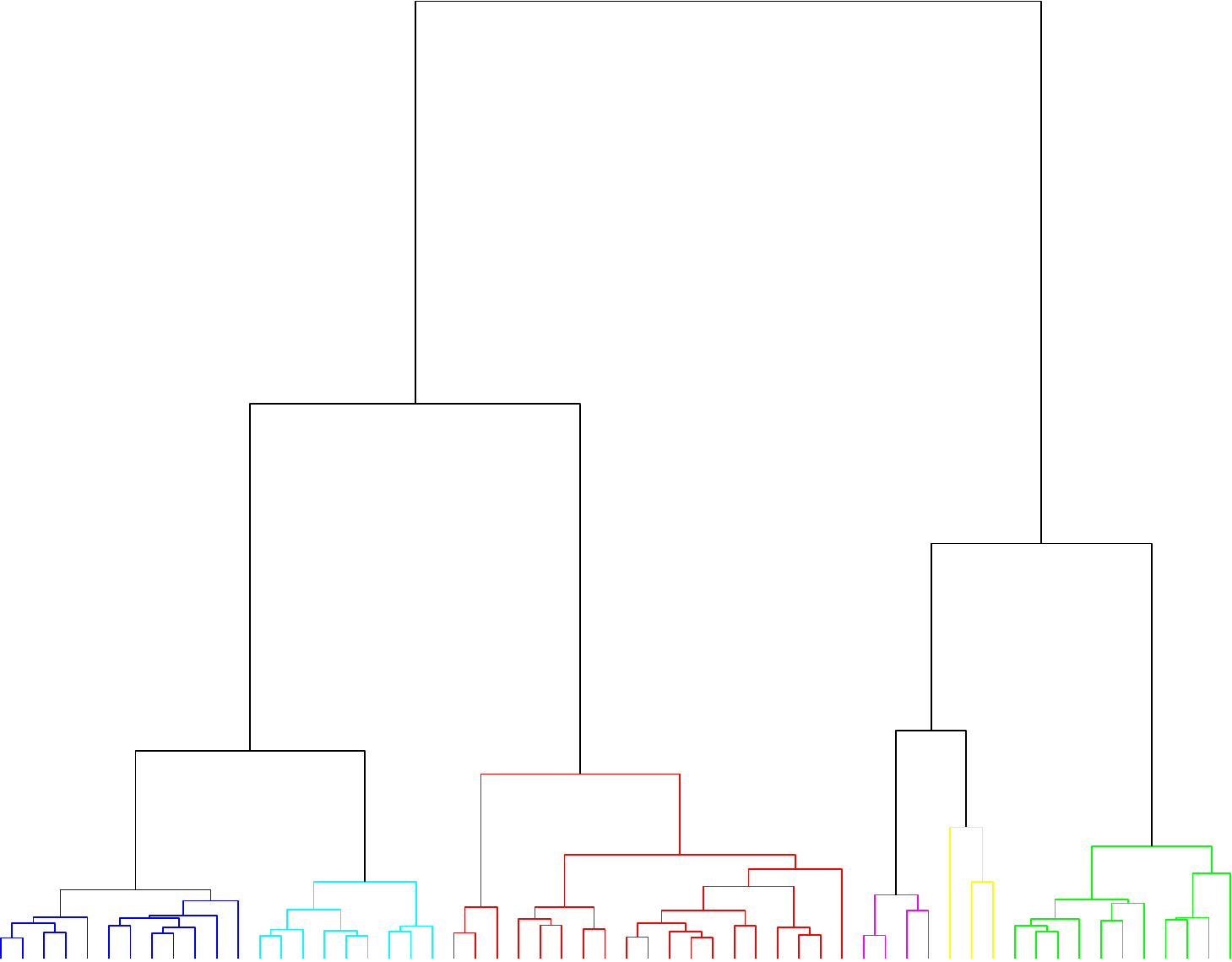}
    \caption{Clustering dendrogram for the tumor dataset with low survival cluster tumors removed.}
    \label{fig:tumor_dendrogram_6}
\end{figure}

\begin{table}[!t]
\caption{Clusterwise mean and median survival for Figure \ref{fig:tumor_dendrogram_6}.}
\label{tab:exp3}
\begin{center}
\begin{small}
\begin{tabular}{|c|c|c|c|c|c|c|}
\hline
&Blue&Cyan&Red&Magenta&Yellow&Green\\
\hline
Mean & 18.1 & 28.0 & 17.9 & 19.4 & 5.0 & 12.6 \\
\hline
Med. & 14.9 & 22.3 & 14.3 & 20.4 & 4.5 & 10.7 \\
\hline
\end{tabular}
\end{small}
\end{center}
\end{table}

\begin{figure}[!t]
    \centering
    \includegraphics[width = 0.09\textwidth]{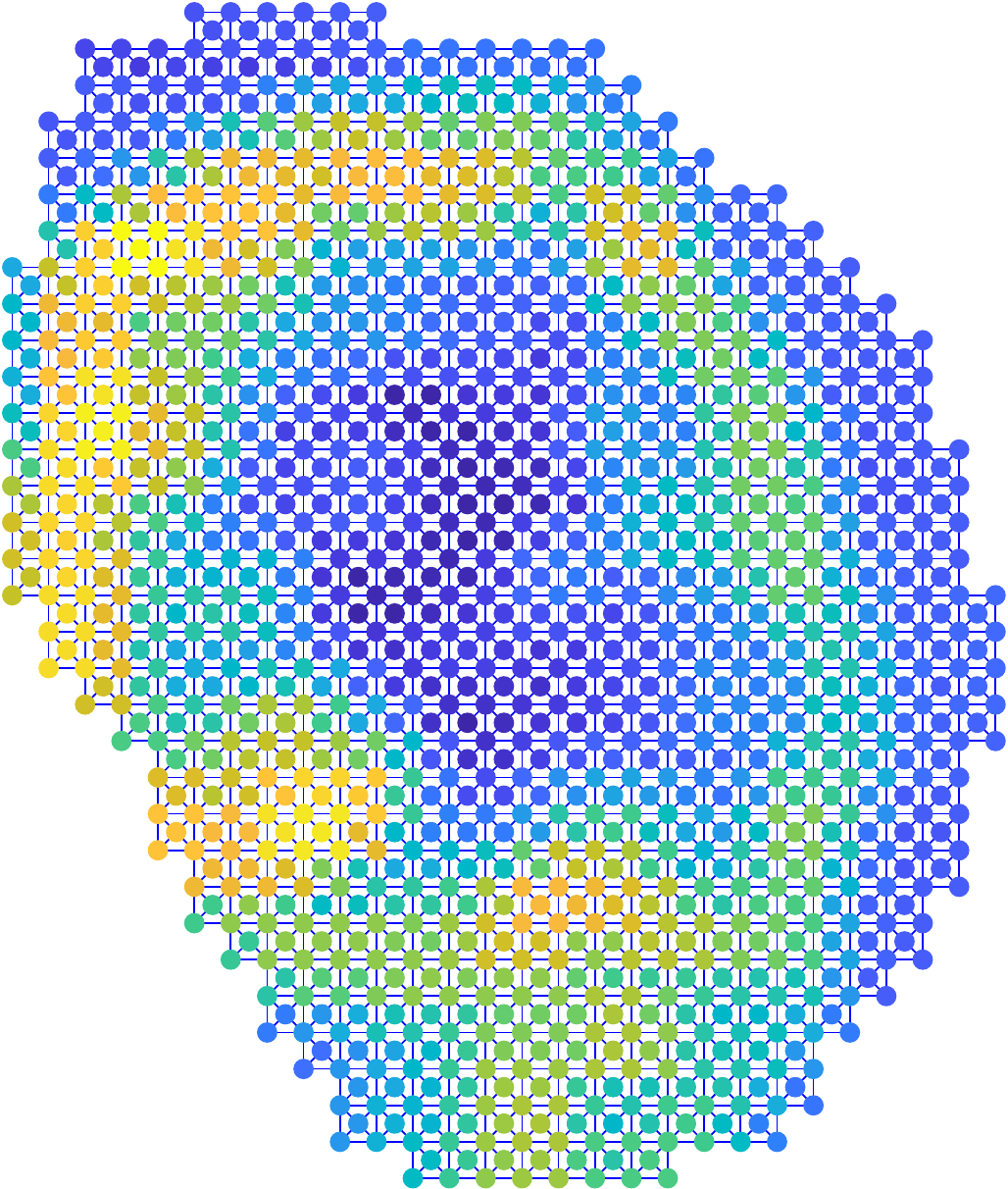}
    \includegraphics[width = 0.09\textwidth]{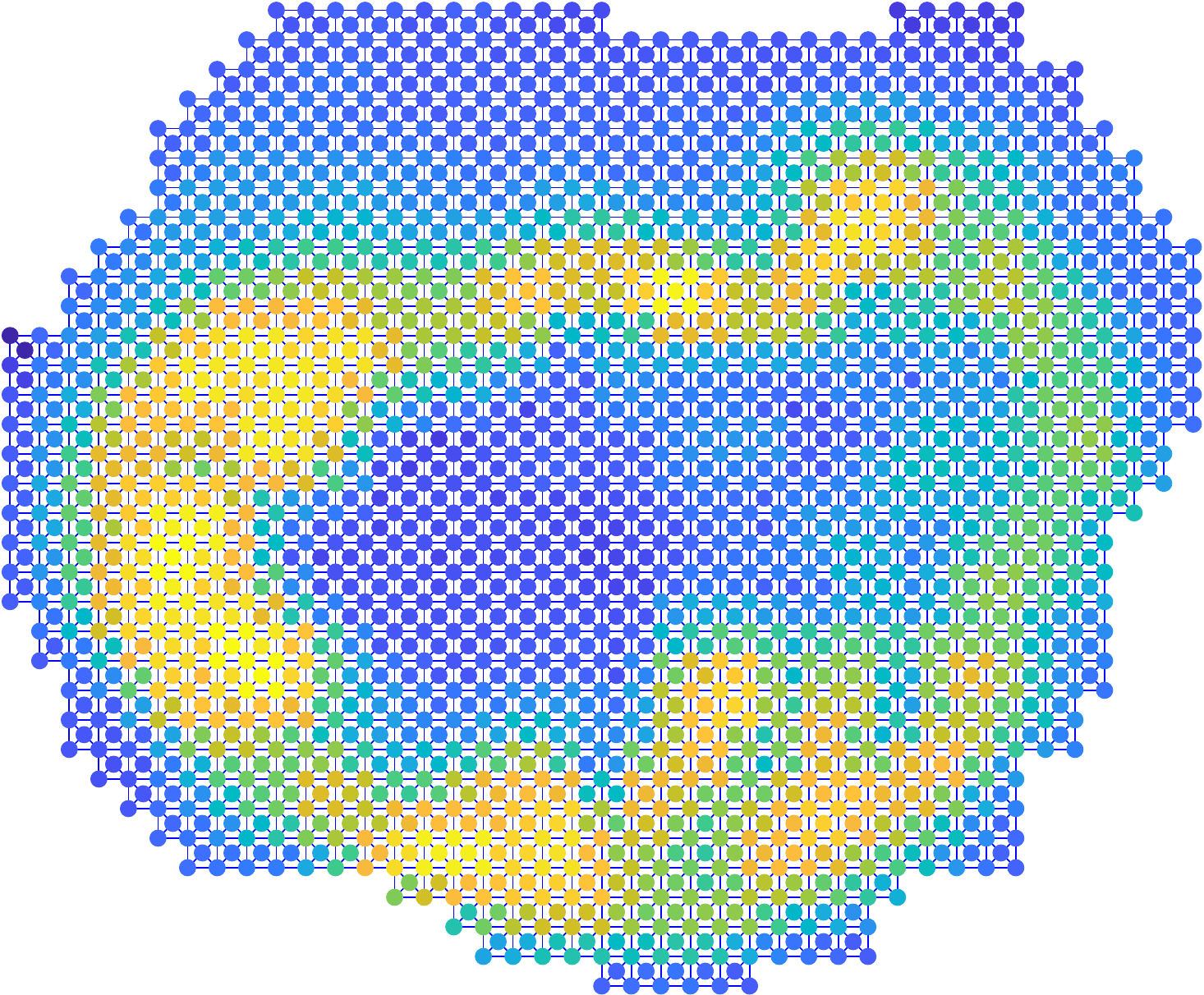}
    \includegraphics[width = 0.09\textwidth]{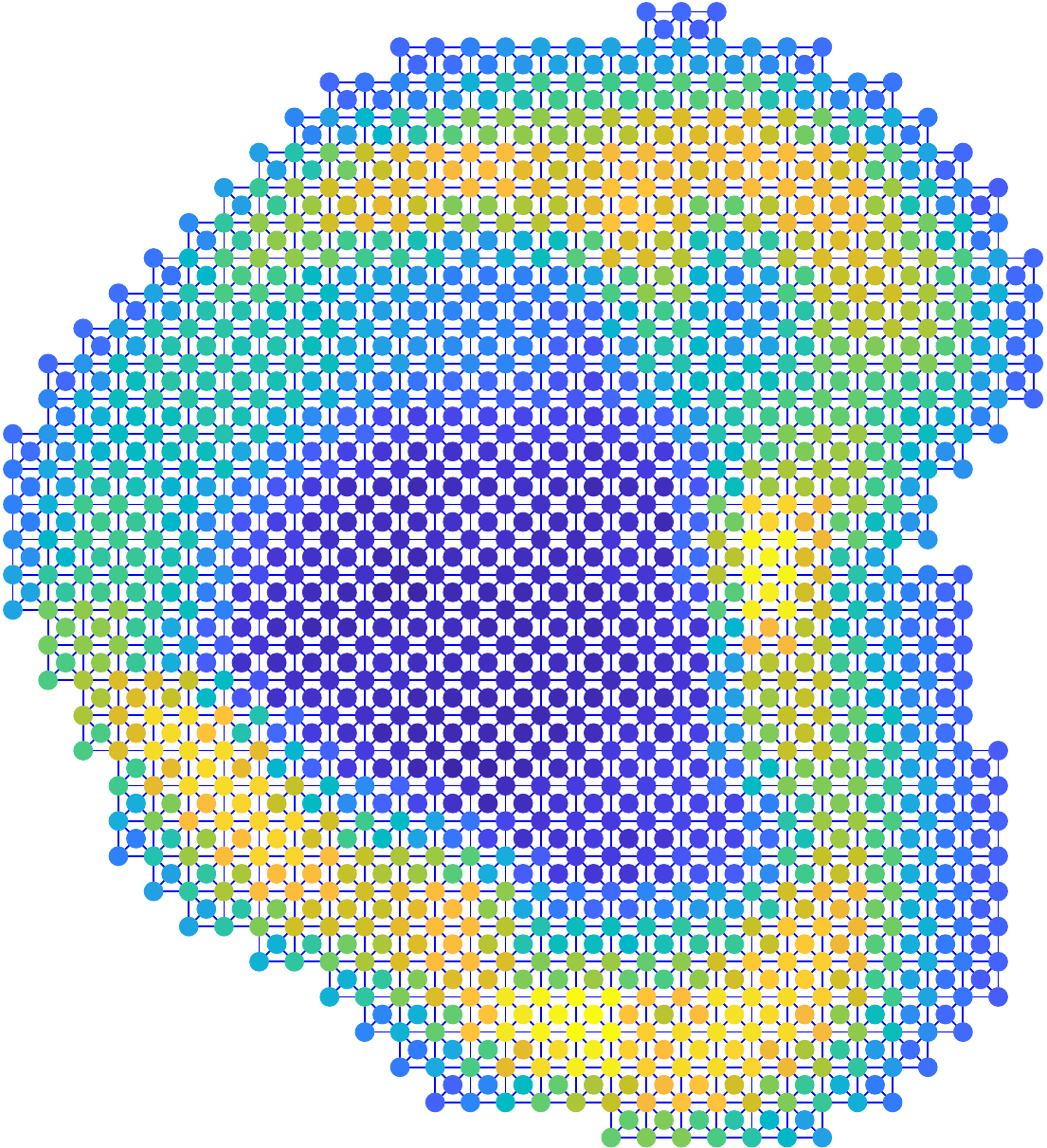}
    \includegraphics[width = 0.09\textwidth]{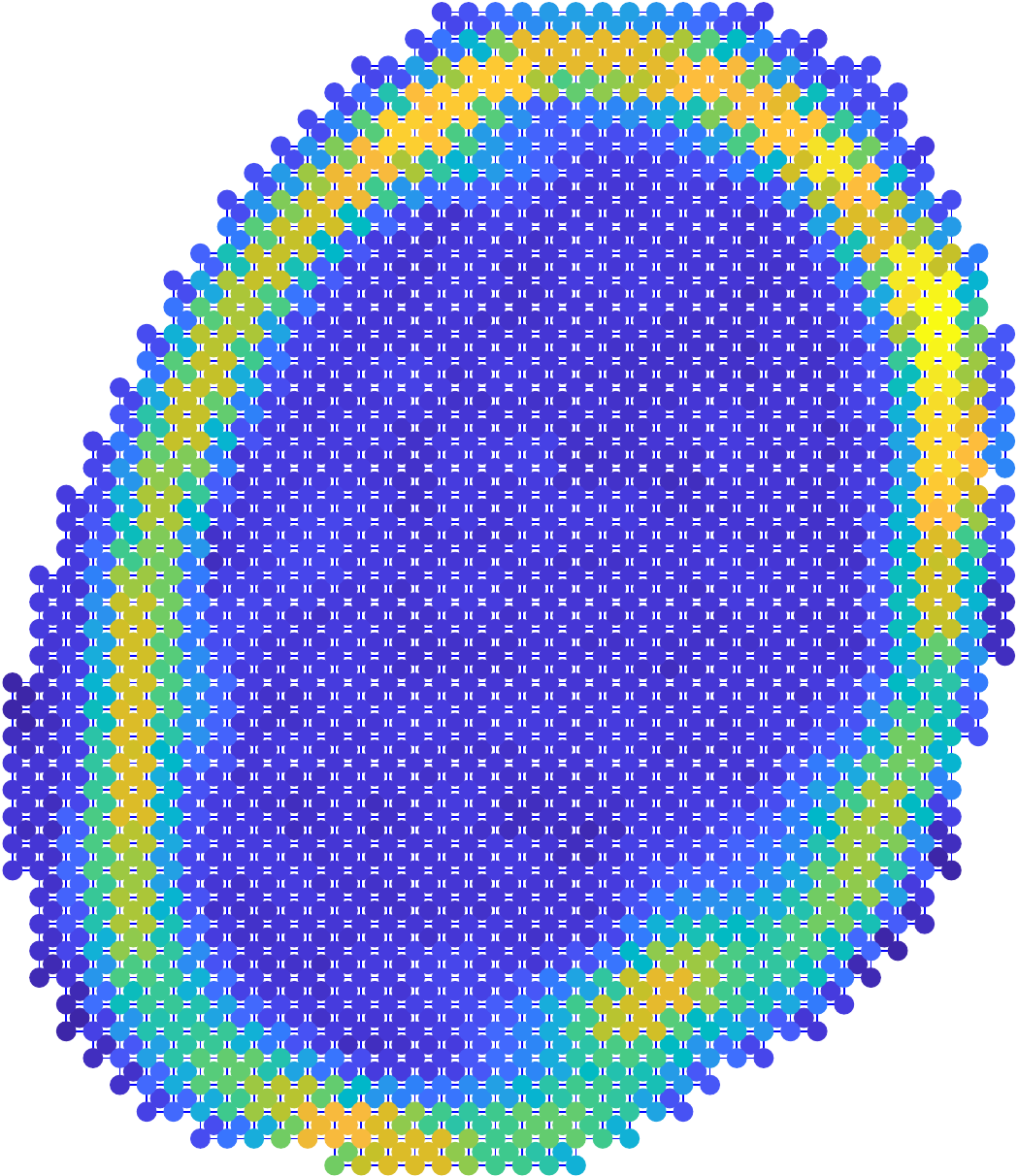}
    \includegraphics[width = 0.09\textwidth]{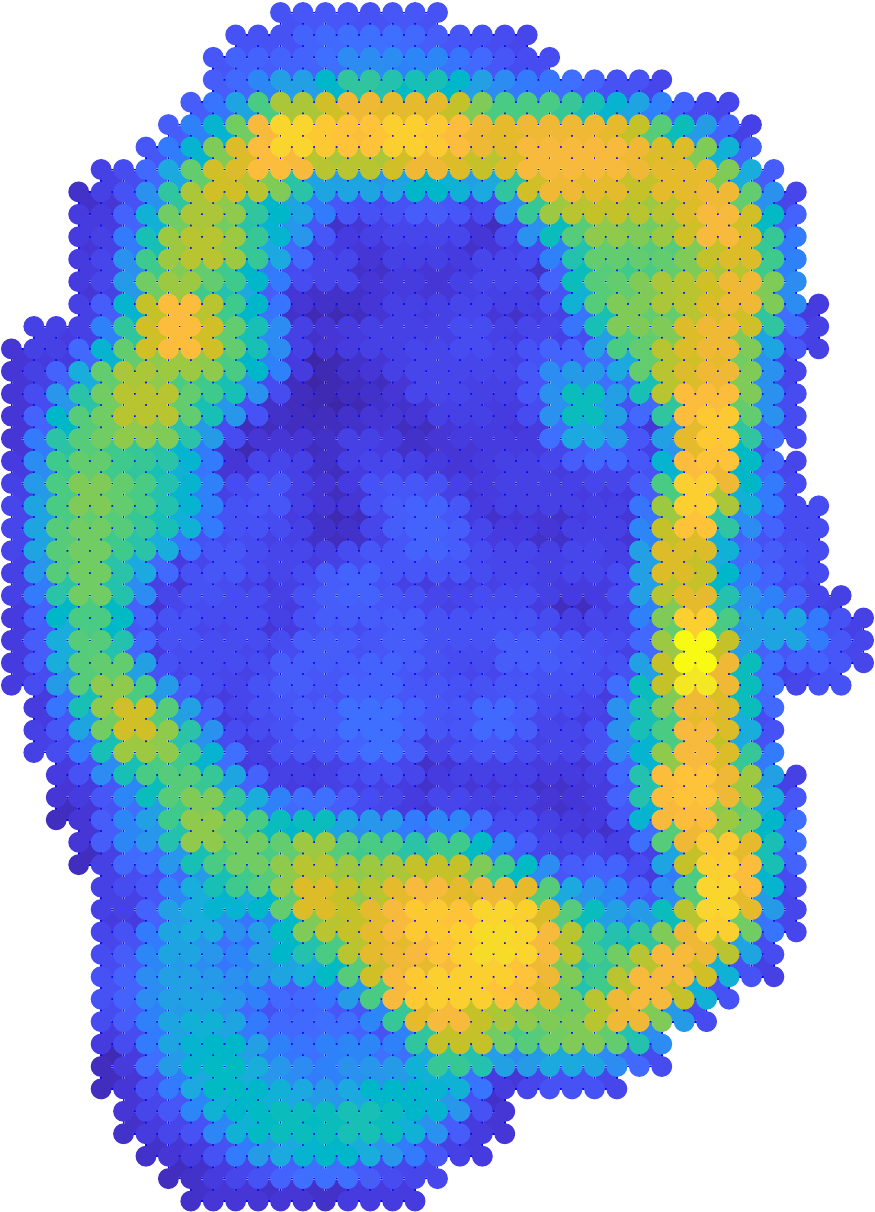}
    \includegraphics[width = 0.09\textwidth]{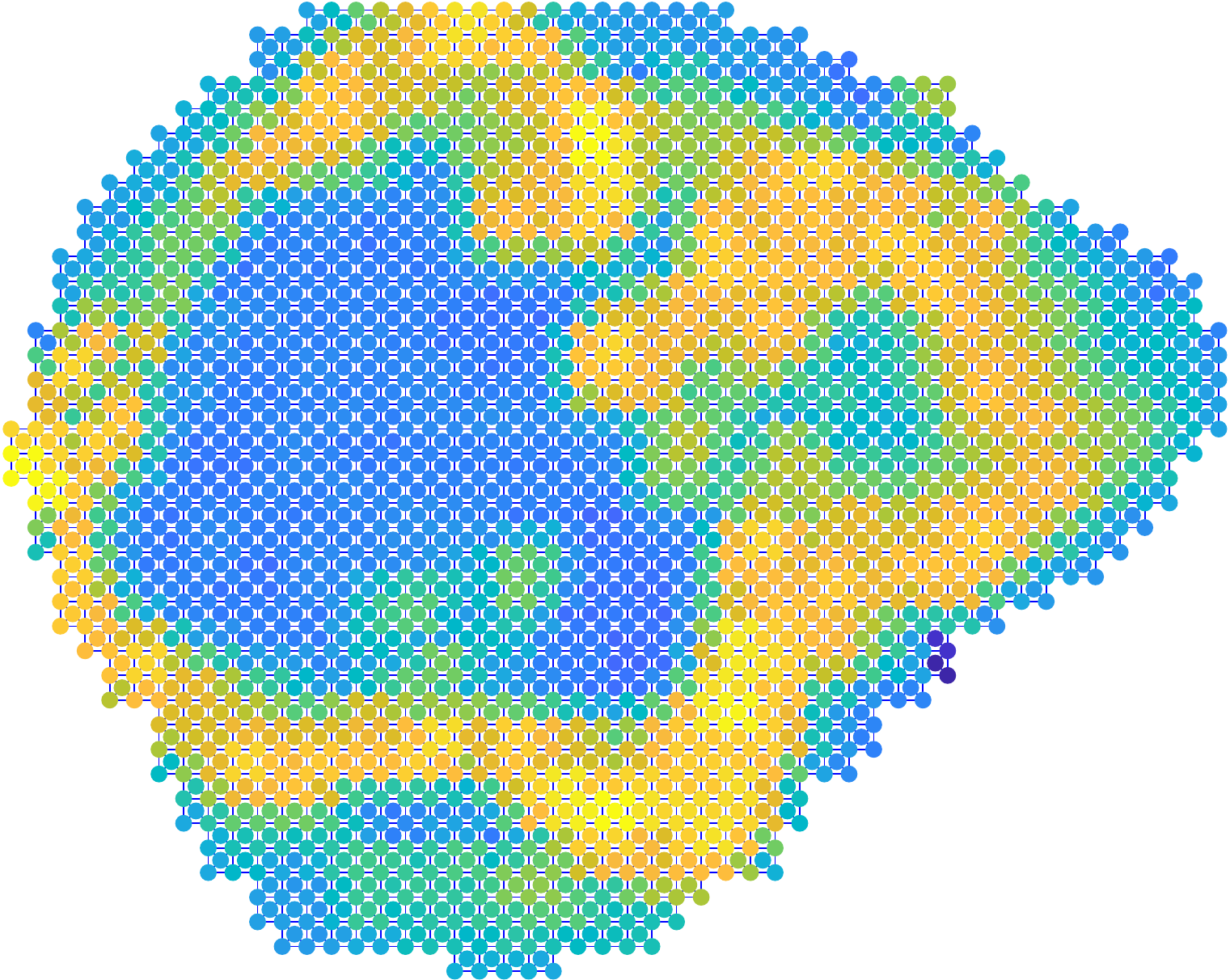}
    \includegraphics[width = 0.09\textwidth]{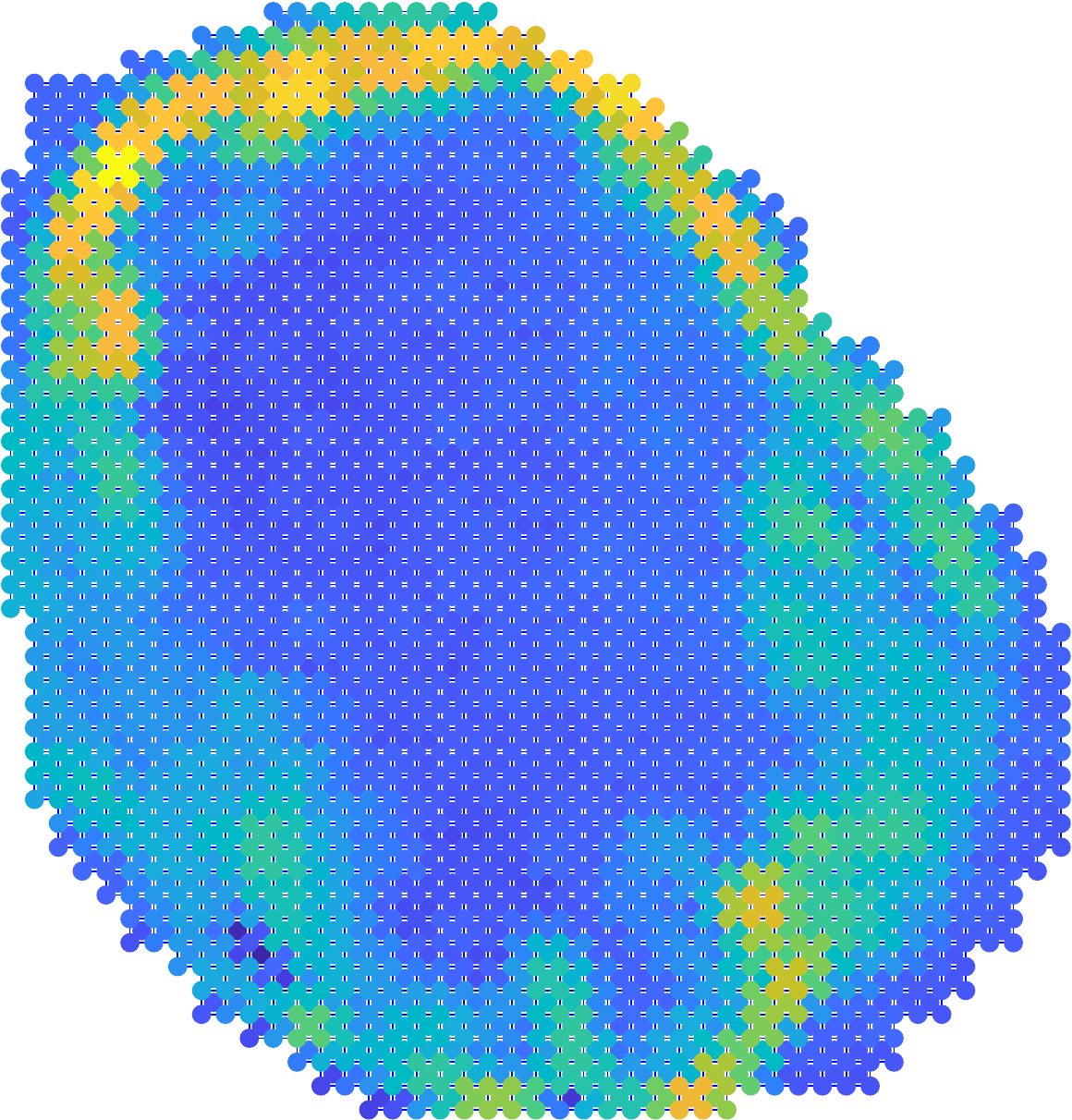}
    \includegraphics[width = 0.09\textwidth]{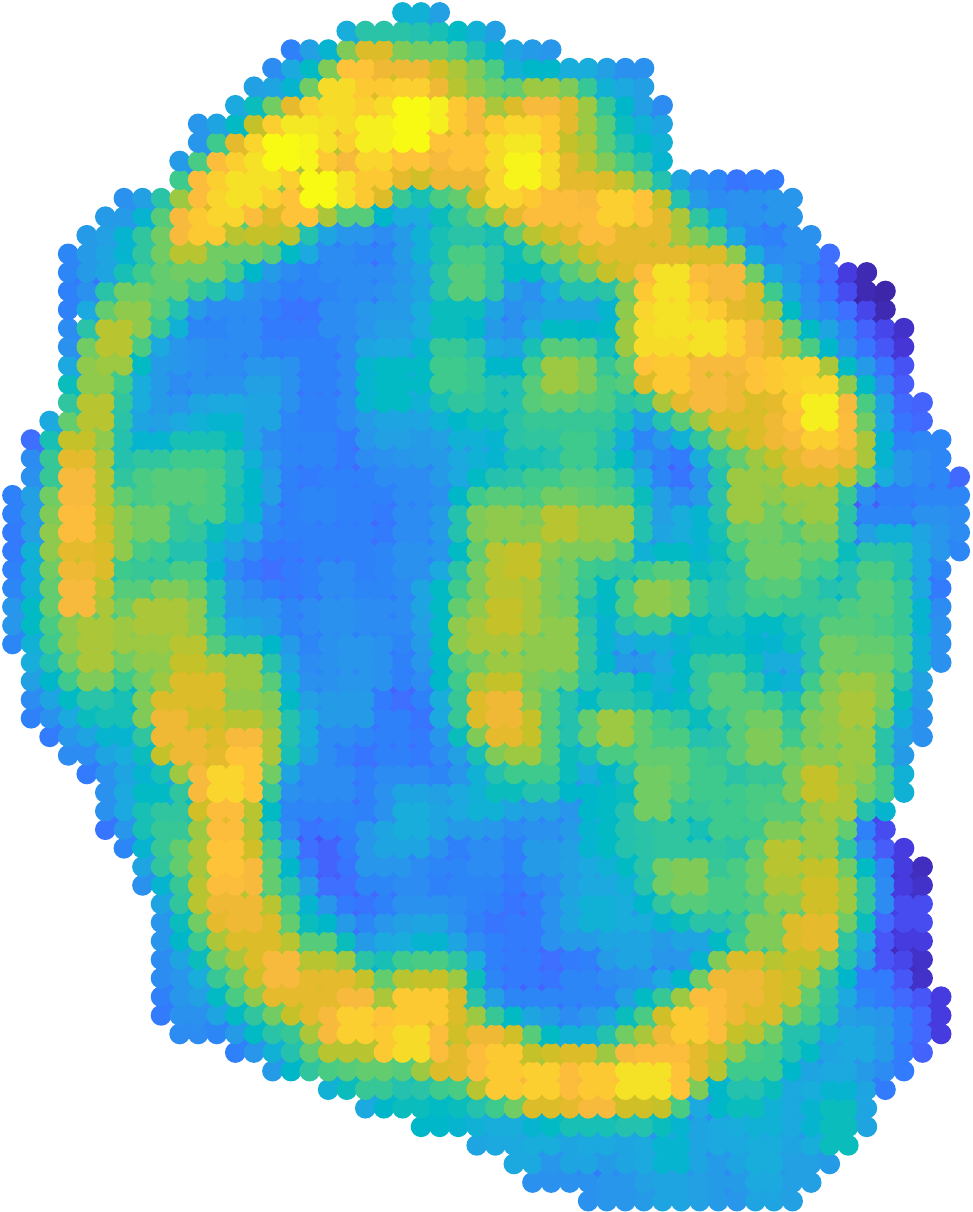}
    \includegraphics[width = 0.09\textwidth]{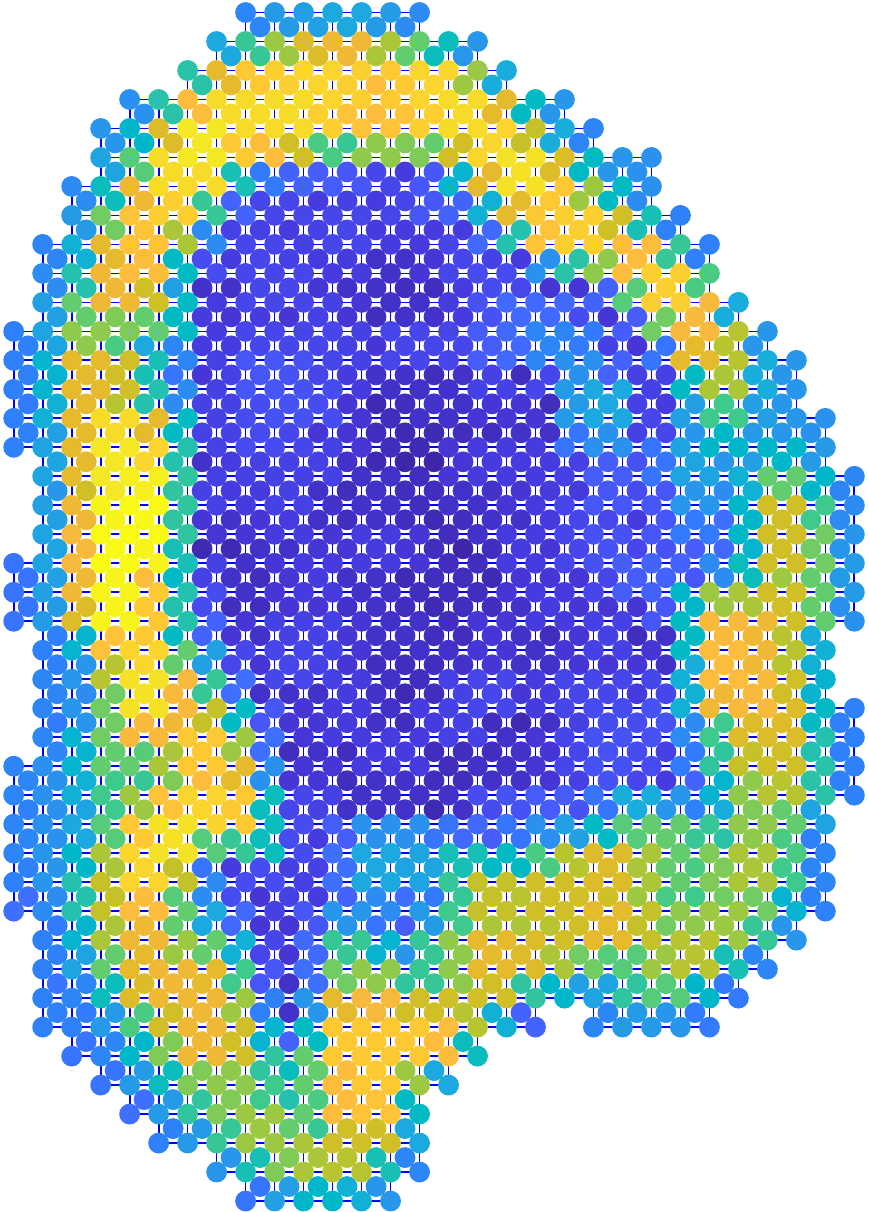}
    \includegraphics[width = 0.09\textwidth]{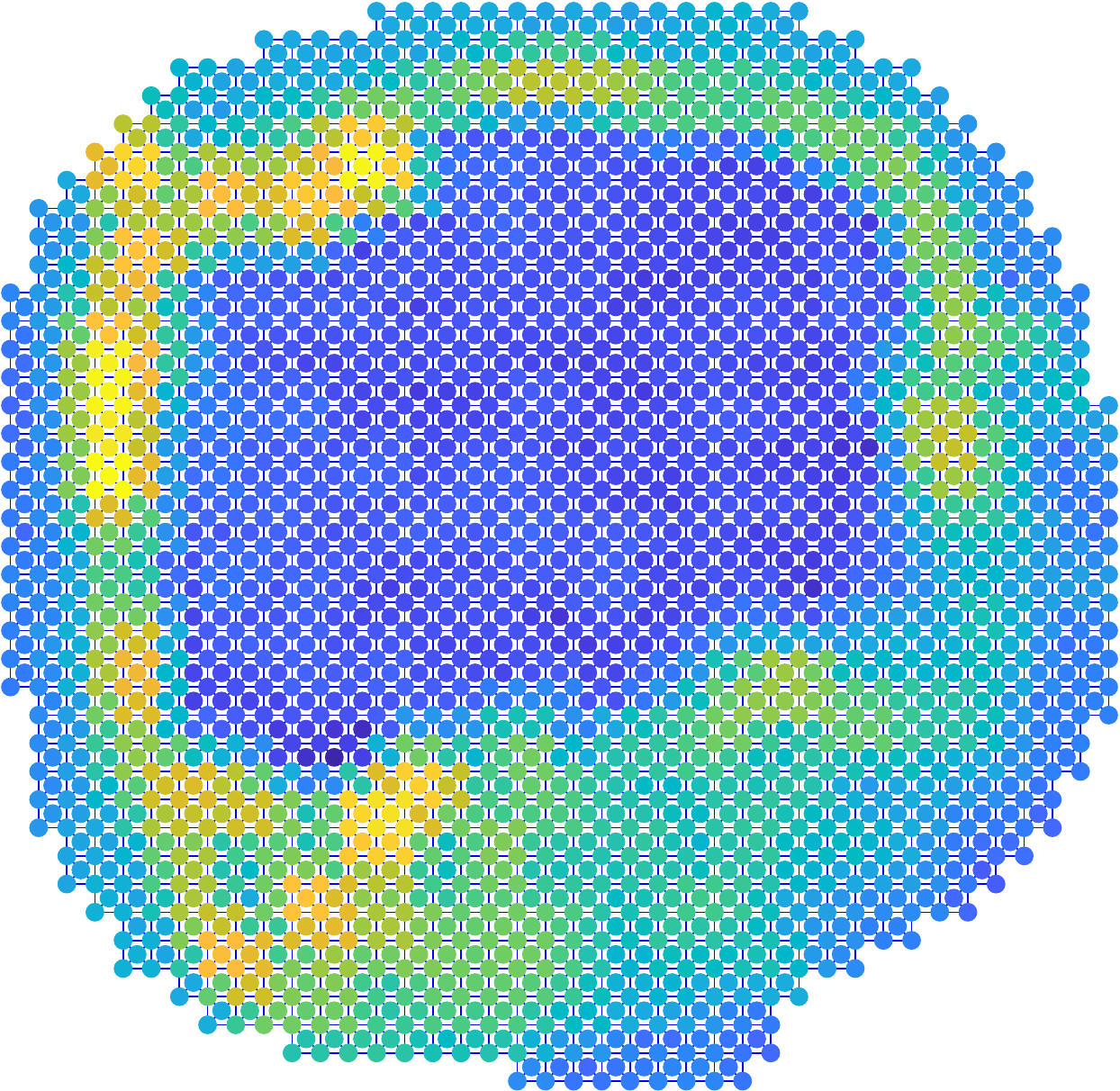}
    \caption{Samples of weighted simplicial complex representations of tumors from cyan (top) and blue (bottom) clusters of Figure \ref{fig:tumor_dendrogram_6}.}
    \label{fig:tumor_cluster_cyan_and_blue}
\end{figure}

\section{Future Work}\label{sec:conclusion}

Our work suggests several directions for future research. Driven by the qualitative distance-based clustering results presented here, we next plan to incorporate WECT representations into more sophisticated statistical models for tumor survival prediction. The WECT representation is flexible in the sense that it provides a summary of any weighted simplicial complex. We plan to apply this type of analysis to other shape data, such as weighted simplicial complexes representing annotated molecule shapes. On the theoretical side, there are several interesting questions left open. Principally, one could attempt to strengthen Theorem \ref{thm:injectivity_theorem} on injectivity of the WECT in several ways. In its current form, it is mainly a theoretical result and an implementation of an inversion algorithm would be desirable. A practical version of such a construction would only require information about weighted Euler curve measurements in finitely many directions, along the lines of results in \cite{curry2018many} on the ECT. It would also be interesting to have a quantitative version of the injectivity theorem; if WECTs of $(K_1,g_1)$ and $(K_2,g_2)$ are close in $L^2$ distance, does this imply that $(K_1,g_1)$ and $(K_2,g_2)$ are close in some resonable metric, such as Wasserstein distance (treating a normalization of $g_j$ as a probability measure supported on $K_j$)?

\noindent\textbf{Acknowledgments:} We thank Arvind Rao for sharing the GBM dataset. SK was partially supported by NSF DMS-1613054, NSF CCF-1740761, NSF CCF-1839252 and NIH R37-CA214955.

\pagebreak

{\small
\bibliographystyle{ieee_fullname}
\bibliography{egbib}
}

\end{document}